\documentclass[12pt]{article}
\usepackage{setspace}
\usepackage{natbib}
\usepackage{graphicx}
\usepackage{lscape}
\usepackage{pdflscape}
\usepackage{afterpage}
\usepackage{pdfpages}
\usepackage{float}
\usepackage{booktabs}
\usepackage[utf8]{inputenc}
\usepackage{indentfirst}
\usepackage{arydshln}
\usepackage{arydshln}
\usepackage{tikz}
\usepackage{lipsum}
\usepackage{pgffor}
\usepackage{dsfont}
\usepackage{amsfonts}
\usepackage{amsmath}
\usepackage{xfrac}
\usepackage{amssymb}
\usepackage{graphicx}
\usepackage{url}				
\usepackage{subfig} 
\usepackage{bbm}

\usepackage{multirow}

\newtheorem{ass}{Assumption}[section]

\newtheorem{cor}{Corollary}[section]

\newtheorem{prop}{Proposition}[section]

\newtheorem{remark}{Remark}

\newenvironment{proof}[1][Proof]{\noindent \textbf{#1.} }{\  \rule{0.5em}{0.5em}}

\usepackage{hyperref} 
\hypersetup{
	colorlinks=true, breaklinks=true, bookmarks=true,bookmarksnumbered,
	urlcolor=red, linkcolor=red, citecolor=blue, 
	pdftitle={}, 
	pdfauthor={\textcopyright}, 
	pdfsubject={}, 
	pdfkeywords={}, 
	pdfcreator={pdfLaTeX}, 
	pdfproducer={LaTeX with hyperref and ClassicThesis} 
}

\usepackage[top=1in, bottom=1in, left=1in, right=1in]{geometry}

 \usepackage{clipboard}
\newclipboard{JAE_revision}
\openclipboard{JAE_revision}

\begin{document}

	\def\spacingset#1{\renewcommand{\baselinestretch}%
		{#1}\small\normalsize} \spacingset{1}

	
\title{ {
\LARGE Inference in Difference-in-Differences: How Much Should We Trust in Independent Clusters?} \footnote{I would like to thank Xavier D'Haultfoeuille, Vitor Possebom, Pedro Sant'Anna, Jon Roth, Elie Tamer and  Matthew Webb   for comments and suggestions.  Luis Alvarez, Lucas Barros, Raoni Oliveira, and Flavio Riva provided exceptional  research assistance.  I  gratefully acknowledge financial support from FAPESP and CNPq.  }}

\author{
Bruno Ferman\footnote{email: bruno.ferman@fgv.br; address: Sao Paulo School of Economics, FGV, Rua Itapeva no. 474, Sao Paulo - Brazil, 01332-000; telephone number: +55 11 3799-3350}  \\
\\
Sao Paulo School of Economics - FGV \\
\\
\footnotesize
First Draft: May 8th, 2019 \\
\footnotesize
This Draft: September 5th, 2022
}

\date{}
\maketitle

	\newsavebox{\tablebox} \newlength{\tableboxwidth}
	

	\begin{center}

\href{https://sites.google.com/site/brunoferman/research}{Please click here for the most recent version}

\

\

\textbf{Abstract}

\end{center}

We analyze the challenges for inference in  difference-in-differences (DID) when there is spatial correlation. We present novel theoretical insights and empirical evidence on the settings in which ignoring spatial correlation should lead to more or less distortions in DID applications. We show that details such as the time frame used in the estimation, the choice of the treated and control groups, and the choice of the estimator, are key determinants of  distortions due to spatial correlation.  We also analyze the feasibility and trade-offs involved in a series of  alternatives to take spatial correlation into account. Given that, we provide relevant recommendations for applied researchers on how to mitigate and assess the possibility of inference distortions due to spatial correlation.

\

	\noindent%
	{\it Keywords:}  spatial correlation; clustered standard errors; linear factor model
	
	\
	
	\noindent%
	{\it JEL Codes:} C12; C21; C23; C33 
	
		\vfill

	\newpage
	\spacingset{1.45} 
	

\doublespace

\section{Introduction}

Difference-in-Differences (DID) is one of the most widely used methods for identification of causal effects in social sciences. However, inference in DID  can be complicated by both serial and spatial correlations.  \Copy{Cluster}{ After an influential  paper by \cite{Bertrand04howmuch}, showing that serial correlation can lead to severe over-rejection in DID applications, most papers applying DID use inference methods that are robust to arbitrary forms of serial correlation. A common alternative in this case is to rely on cluster robust variance estimator (CRVE) at  the unit level, which allows for arbitrary serial correlation, but generally relies on the assumption  that these unit-level clusters are independent.  In most cases, DID papers do not take the possibility of  spatial correlation across these clusters into account.\footnote{In case we have, for example, individual-level data and a state-level policy,  clustering at the state level would allow for arbitrary correlation between individuals in the same state. Since clustering at the state level  takes within-state correlation into account, we focus on the possibility of across state spatial correlations.   } } 

While there are some alternatives for inference in the presence of spatial correlation, they generally require knowledge about the relevant distance metric, impose assumptions on the serial correlation, and/or rely on more data (such as a large number of periods).\footnote{For example, \cite{CONLEY19991}, \cite{KIM201385}, \cite{CT} (in their online appendix A.3),  \cite{BESTER2011137}, and \cite{Muller1,Muller2} rely on distance measures across units.  Other papers exploit the time dimension to perform inference in the presence of spatially correlated shocks. However, these methods rely on a large number of periods (for example,  \cite{VOGELSANG2012303},  \cite{FP} (Section 4) and \cite{Chernozhukov}).  \cite{ferman2020inference} considers a setting in which spatial dependence is unknown, and the number of pre-treatment periods is fixed. However, his conclusions rely on a strong mixing condition for the spatial correlation, and are only valid for settings with few treated and many control units.  \label{footnote}}   
We discuss different settings in which such alternatives are unfeasible in DID applications. As we illustrate below, this may be the case when the relevant source of spatial correlation is unknown by the applied researcher. Moreover, even  when the relevant distance metric is known, there might  not be enough variation in the data to estimate the  spatial correlation. We also show that, even when feasible, some alternatives might have  important limitations that have not been previously considered in the literature.\footnote{For example, while the common correlated effects estimator proposed by \cite{Pesaran} may provide an interesting alternative in some settings, we discuss some limitations in Appendix \ref{Appendix_LFM}. We also discuss in Appendix \ref{Appendix_alternatives_CPS} the use of the inference method proposed by \cite{Muller1,Muller2}, for settings in which the distance metric is known, but there is limited variation on the distance metric in the data.  }

Given that correcting for spatial correlation  is not always feasible, we consider the consequences of ignoring spatial correlation in DID applications.  We analyze a setting in which the spatial correlation follows a linear factor model, which allows for a rich variety of spatial correlation structures.  We show in Section \ref{problem} that, in a setting with no variation in treatment timing,  inference ignoring spatial correlation becomes more problematic when \emph{both}  (i) the variance of the difference between the pre- and post-treatment averages of common shocks is large relative to the variance of the same difference for the idiosyncratic shocks, \emph{and} (ii) the distribution of factor loadings has different expected values for treated and control units. When at least one of these conditions does not hold, the time and/or unit fixed effects would absorb most of the relevant spatial correlation. This provides novel insights on the settings in which spatial correlation should lead to more or less distortions for inference in DID applications. In particular, we show that details such as the time frame used in the estimation, the choice of the treated and control groups, and the choice of the estimator, are key determinants of  distortions due to spatial correlation.
 
We then present in Section \ref{simulations} two sets of simulations, based on the American Community Survey (ACS) and  the Current Population Survey (CPS). In the first one, we illustrate a setting in which  the source of spatial correlation is unknown to the applied researcher. In the second one, the source of spatial correlation is known, but there is  not enough variation in the data to take that into account. In both cases, ignoring spatial correlation does not significantly affect inference when the time frame is short, but can lead to relevant size distortions when the time frame is long. Based on our theoretical results, this is consistent with common shocks also being more serially correlated relative to the idiosyncratic shocks. In the second setting, it is also possible to ameliorate the spatial correlation problems by considering treated and control groups that are more alike. This is also consistent with our theoretical results. 
 Section  \ref{recommendations} concludes with  recommendations for applied researchers.

\section{The Inference Problem} \label{problem}

We start presenting in Section \ref{simple_DID}  a general DID model in which we discuss the consequences of ignoring spatial correlation for inference. In Section \ref{Existing_methods}, we  discuss existing alternatives to take spatial correlation into account, and explain the reasons why correcting for spatial correlation may be unfeasible in some settings. In Section \ref{sec_LFM}, we impose more structure on the spatial correlation, so that we can provide further insights on the settings in which we should expect spatial correlation to lead to more or less  distortions for inference. Throughout, we consider the case in which a parallel trends assumption remains valid, so that the inclusion of spatial correlation does \emph{not} imply that the DID model is misspecified.

\subsection{A simple DID model with spatial correlation} \label{simple_DID}

We start considering a standard model for the potential outcomes.  Let $Y_{jt}(0)$ ($Y_{jt}(1)$) be the potential outcome of unit $j$ at time $t$ when this unit is untreated (treated) at this period. We consider first that potential outcomes are given by
\begin{eqnarray} \label{simple_did_equation}
\begin{cases} Y_{jt}(0) =   \theta_j + \gamma_t + \eta_{jt}  \\  Y_{jt}(1) = \alpha_{jt} +Y_{jt}(0),  \end{cases} 
\end{eqnarray}
where $\theta_j$ and $\gamma_t$ are, respectively, unit- and time-invariant unobserved variables, while $\eta_{jt}$ represents unobserved variables that may vary at both dimensions.  We do not impose any restriction on the serial and spatial correlations of $\eta_{jt}$, so this is a very general model; $\alpha_{jt}$ is the (possibly heterogeneous) treatment effect on unit $j$ at time $t$.\footnote{All results  remain unchanged in case we consider a setting with individual-level observations $i$ within units $j$, $Y_{ijt}$, and we consider clustering at the unit level. Within-unit spatial correlation can be taken into account by clustering at the unit level, so we are mainly concerned about the possibility of across-units spatial correlation. }

Equation \ref{simple_did_equation}  leads to a standard DID model for $Y_{jt} = d_{jt}Y_{jt}(1) + (1-d_{jt}) Y_{jt}(0)$, given by
\begin{eqnarray}
Y_{jt} =  \bar \alpha d_{jt} +  \theta_j + \gamma_t + \widetilde \eta_{jt},
\end{eqnarray}
where $\bar \alpha$ is defined as the two-way fixed effect (TWFE) estimand, $ \widetilde \eta_{jt} =  \eta_{jt} +  (\alpha_{jt} - \bar \alpha) d_{jt}$, and $d_{jt}$ is an indicator variable equal to one if unit $j$ is treated at time $t$, and zero otherwise. 

Consider  a simpler case in which  $d_{jt}$ changes to 1 for all treated units starting after date $t^\ast$, and  define a dummy variable $D_j$ equal to one if unit $j$ is treated. There are $N_1$ treated units, $N_0$ control units, and $T$ time periods. Let $\mathcal{I}_1$ ($\mathcal{I}_0$) be the set of  treated (control) units, while  $\mathcal{T}_1$ ($\mathcal{T}_0$) be the set of post- (pre-) treatment periods. For a generic variable $A_t$, define $\nabla A = \frac{1}{T-t^\ast} \sum_{t \in \mathcal{T}_1} A_{t} -  \frac{1}{t^\ast} \sum_{t \in \mathcal{T}_0} A_{t}$.  In particular,  we consider $W_j = \nabla \widetilde \eta_{j}$, which is the post-pre difference in average errors for each unit $j$. 

 In this section, we consider a repeated sampling framework over the distribution of $\{W_j \}_{j \in \mathcal{I}_0 \cup \mathcal{I}_1 }$, conditional on $\mathbf{D} = \mathbf{d}$, where $\mathbf{D} = (D_1,...,D_N)$. For example, we can think of $W_j$ as a linear combination of  economic or weather  shocks that may affect unit $j$, and we analyze the distribution of the DID estimator over the distribution of those shocks.  In this case, we have that $\bar \alpha  = \mathbb{E} [ \frac{1}{N_1}  \frac{1}{T-t^\ast}\sum_{j \in \mathcal{I}_1} \sum_{t \in \mathcal{T}_1}\alpha_{jt} | \mathbf{D} = \mathbf{d} ]$.

In this  setting,  the DID estimator is the same as the TWFE estimator, which is given by
\begin{eqnarray}  \label{alpha}
\hat \alpha  &=& \frac{1}{N_1} \sum_{j \in \mathcal{I}_1} \nabla Y_{j}  - \frac{1}{N_0} \sum_{j \in \mathcal{I}_0} \nabla Y_{j}=
\bar \alpha +  \frac{1}{N_1} \sum_{j \in \mathcal{I}_1} W_j - \frac{1}{N_0} \sum_{j \in \mathcal{I}_0} W_j.
\end{eqnarray}

 \Copy{W}{If we have $\mathbb{E}[W_j | \mathbf{D} = \mathbf{d}] =0$ for all  $j$, then the DID estimator $\hat \alpha$ will be  unbiased for $\bar\alpha$, regardless of the assumptions on the serial and spatial correlations of $\widetilde \eta_{jt}$. }
However,  inference  is only possible if we impose assumptions on either the serial or the spatial correlation of $\widetilde \eta_{jt}$. Most commonly, inference methods for DID do not impose restrictions on the serial  correlation of $\widetilde  \eta_{jt}$, but assumes  that $\widetilde \eta_{jt}$ are independent across $j$.\footnote{See, for example, \cite{Arellano}, \cite{Bertrand04howmuch}, \cite{cameron2008bootstrap},  \cite{IZA},  \cite{CT},  \cite{FP},   \cite{Canay}, and \cite{MW}. }  A common alternative in this case is to rely on CRVE at the unit level which, assuming independence across $j$,  is  valid  when both $N_1$ and $N_0$ are large.

Now consider a setting in which  treatment allocation is such that units that are exposed to similar shocks are also more likely to be allocated into the same treatment status. Then, once we condition on $\mathbf{D} = \mathbf{d}$, we should expect a strong correlation  between $W_j$ and $W_{j'}$ if  $j$ and $j'$ received the same treatment allocation. In such cases,   not taking such spatial correlation into account can lead to  over-rejection.   The intuition  is the following. Imagine there is an unobserved variable in $W_j$ that equally affects all treated units, but does not affect the control units.\footnote{We assume that the expected value of this variable is equal to zero conditional on $\mathbf{D}= \mathbf{d}$, so the presence of such correlated shock does not affect the identification assumption of the DID model.}  If the null $H_0: \bar \alpha=0$ is true, then  $\hat \alpha = \frac{1}{N_1} \sum_{j \in \mathcal{I}_1} W_j - \frac{1}{N_0} \sum_{j \in \mathcal{I}_0} W_j$.  Therefore, under the null,  finding a ``large'' value for $\hat \alpha$ would only be possible if many of those $W_j$ for $j \in  \mathcal{I}_1$ were positive, and/or many of those $W_j$ for $j \in  \mathcal{I}_0$ are negative. We would consider that this event has a lower probability than the true one if we (mistakenly) assume that $W_j$ are independent, leading to over-rejection.

\subsection{Existing solutions and their limitations} \label{Existing_methods}

There are  alternatives for inference when we relax the assumption that clusters are independent. However, such alternatives generally assume that there is a distance metric across units, impose assumptions on the serial correlation, and/or rely on more data,  (such as a large number of periods).\footnote{See Footnote \ref{footnote}.}

We focus on settings in which alternatives to take spatial correlation into account may be unfeasible. For example, it may be that  the relevant source of spatial correlation is unknown to the econometrician. While it is natural to think about spatial correlation considering geographical distances, the relevant spatial correlation may arise from other sources. For example, we consider in Section \ref{Sim_ACS} simulations in which  PUMA's with some specific  industry compositions are more likely to receive  treatment. In such case, even if we assume conditions such that the DID estimator is unbiased, ignoring spatial correlation from unobserved shocks that are related to industry composition might generate relevant size distortions. 
Moreover, attempts to correct for that considering that the relevant distance metric is geographical would generally not solve the problem.  

There are  alternatives that take spatial correlation into account even when the source of spatial correlation is unknown, exploiting the time series of the data.\footnote{See Footnote \ref{footnote}. Also, considering the use of two-way cluster at the unit and time dimensions would not provide a valid solution in this setting, even if  both $N$ and $T$ are large, because it would not take into account the correlation between $\eta_{jt}$ and $\eta_{j't'}$, for $j \neq j'$ and $t \neq t'$. See \cite{multi_way}, \cite{THOMPSON20111}, \cite{Davezies}, \cite{menzel}, and \cite{Webb_multi} for recent developments on multi-way clustering. \label{Footnote_twoway}} However, such alternatives generally require a large number of periods, while, as  \cite{Roth} points out,  settings with short time series  are prevalent in DID applications. One exception  that may take  spatial correlation into account, even when the source of spatial correlation is unknown and $T$ is finite, is the common correlated effects (CCE) estimator, proposed by \cite{Pesaran}. However, we show in Appendix \ref{Appendix_LFM} that there are some limitations and trade-offs involved in using this alternative in our setting. First, it requires variation in treatment timing, so it would not be an option in common settings in which all treated units start treatment at the same time. Also, the CCE estimator imposes restrictions on the dimension of the common shocks.\footnote{We present simulations in Appendix \ref{Appendix_LFM} in which the spatial correlation comes from a linear factor model, as we consider in Section \ref{sec_LFM}.  Inference for CCE leads to large over-rejections when the dimension of the linear factor model is greater than two.} Moreover, in some cases there is a loss in precision relative to considering the TWFE estimator. Finally,  if treatment effects are heterogenous, we show that the CCE estimand may be negative even when treatment effects are always positive.\footnote{This problem has been documented in the DID literature for the TWFE estimator \citep{Bacon,chaisemartin2018twoway}, but not for the CCE estimator. We also show that standard solutions to this problem are unfeasible when we consider the CCE estimator.}

Moreover, even if the source of spatial correlation is known, it might be unfeasible to take the spatial correlation into account. This may happen when we do not have enough variation in the data to estimate the relevant spatial correlations.  As an example, suppose we have data on students' test scores for grades one and two, and we have a policy that affected only second graders  in the post-treatment periods. In this case, there might be common shocks that differentially affect different grades, which might generate relevant spatial correlation for the DID estimator. However, it would unfeasible to cluster at the grade level, or to consider alternative spatial correlation-robust methods, with only two grades.

Finally, we note that a commonly-used rule-of-thumb is to consider CRVE ``at the level of the treatment assignment''.\footnote{See, for example, \cite{NBERw24003} and \cite{MW2020}.} Consider, for example, a setting in which  we analyze a state-level policy, and we have county-level data. In this case, clustering at the county level would generally lead to over-rejection. In contrast, if we have  treatment  \emph{completely randomly assigned} at the state level, then CRVE at the state level would be valid if we have a large number of treated and control states, as  \cite{IK} and \cite{NBERw24003} show considering a design-based approach for inference.\footnote{Following \cite{IK}, we consider that  treatment is ``completely randomly assigned at the state level'' if all possible treatment allocations subject to the constraints on the number of treated states have the same probability.} While we focus in the main text on a setting in which potential outcomes are stochastic, we also consider in Appendix \ref{appendix_design} a design-based approach for inference. 

More generally, however, clustering at the level of the treatment assignment may not solve the  problem in case we have more complex treatment assignments. For example, consider we have two regions, and treatment is assigned based on a  two-stage randomization. First, we have that the proportions of treated states in regions A and B are either $(70\%,30\%)$ or $(30\%,70\%)$, with equal probabilities. Then,  states within each region are randomly allocated into treatment according to those proportions. In this case, the DID estimator is unbiased \citep{rambachan2020designbased}.  In such setting, we have that counties within the same state have the same treatment status. Moreover, in each region we would have both treated and control counties. Therefore, an applied researcher looking at the data might say that  ``treatment is allocated at the state level.'' However,   clustering at the state level would not generally be valid in this case. 

While an alternative in this case would be to cluster at a higher level (in this case, regions), the econometrician may be unaware of this more complex assignment design, and/or not have information to construct the relevant cluster level. Moreover, even if this information is available, when we  consider clustering at higher levels, we may end up with very few clusters to estimate the standard errors. While there are  alternatives  that work in settings with few clusters,\footnote{See, for example, \cite{cameron2015practitioner},  \cite{Ibragimov}, \cite{Canay}, \cite{Hagemann}. } such alternatives generally do  not work well in the limit when we end up with only two or three clusters.

\subsection{A linear factor model for the spatial correlation} \label{sec_LFM}

\subsubsection{Setting}

In order to provide further insights on the implications of spatial correlation, we impose more structure on the errors. We assume that potential outcomes follow a linear factor model
\begin{eqnarray} \label{LFM}
\begin{cases} Y_{jt}(0) = \theta_j + \gamma_t+ \lambda_t \mu_j + \epsilon_{jt}  \\  Y_{jt}(1) = \alpha_{jt} +Y_{jt}(0),  \end{cases} 
\end{eqnarray}
where $\lambda_t$ is an $(1 \times F)$ vector of common shocks, while $\mu_j$ is an $(F \times 1)$ vector of factor loadings  determining how unit $j$ is affected by  $\lambda_t$. While $\theta_j$ and $\gamma_t$ could have  been included as components of $\mu_j$ and $\lambda_t$, we consider them separately to highlight that we can still have time-invariant and unit-invariant shocks as in standard DID model, so what we add is the possibility of other spatially correlated shocks that are not time- nor unit-invariant, which are captured by $\lambda_t \mu_j$.\footnote{We discuss in Appendix \ref{Appendix_LFM} the possibilities of using alternative estimators designed for panel data settings with an error structure following a linear factor model.  }

Such structure allows for a rich variety of spatial correlation structures. We can consider, for example, the case in which spatial correlation comes from counties with similar industry compositions having correlated errors. In this case, we would have $F$ industries, and vector $\mu_j$ would represent the exposure of county $j$ to each of these industries, while $\lambda_t$ would represent industry shocks. {We can also consider the case of $N$ municipalities divided into $F$ states, where there are relevant state-level shocks. In this case, if municipality $j$ belongs to state $f$, we could model that by setting the $f-$th entry of $\mu_j$ equal to one and zero otherwise.}\footnote{This simple structure would not allow for arbitrary spatial correlation within states, as it considers a common state-level shock. We would be able to consider more complex within-state correlations by increasing the dimension of the $\lambda_t$. \label{State-shock}} As another example, this structure can encompass the common notion that spatial correlation depends on geographical distances. Finally, note that the spatial correlation structure may involve different notions of distance (for example, depending on both geographical position and industry composition distances, as considered in the simulations in Section \ref{Sim_ACS}).\footnote{While we focus in the case in which the dimension $F$ is fixed, we consider in Appendix \ref{Appendix_F} a setting in which the dimension $F$ may increase with $N$.  }

We continue to consider that treated units start treatment after $t^\ast$, and let $D_j =1$ if unit $j$ is treated, and $0$ otherwise.  But now we consider the distribution of the DID estimator based on a repeated sampling framework over the distributions of  $D_j$, $\lambda_t$, $\mu_j$, $\epsilon_{jt}$ and $\alpha_{jt}$.  
 
 \begin{ass}{(sampling)}
\normalfont \label{rv}
We observe a sample $\{ Y_{j1},...,Y_{jT}, D_j \}_{j =1}^N$, where $Y_{jt} = D_j Y_{jt}(1) + (1-D_j) Y_{jt}(0)$ if $t>t^\ast$, and $Y_{jt}(0)$ otherwise. Potential outcomes are determined by Equation (\ref{LFM}). The sequence $\{D_j, \mu_j, \epsilon_{j1}, \hdots ,\epsilon_{jT}, \alpha_{jt^\ast+1},...,\alpha_{jT}   \}_{j=1}^N$ is iid, and independent of  $\{\lambda_t\}_{t=1}^T$. $\mathbb{E}[D_j] = c \in (0,1)$, and all random variables have finite variances.
\end{ass}

Assumption \ref{rv} implies that all spatial correlation is captured by this linear factor structure, so that the idiosyncratic shocks $\epsilon_{jt}$ are independent across $j$. We do allow, however, for arbitrary serial correlation in both $\epsilon_{jt}$ and $\lambda_t$.  We also assume for simplicity that treatment effects $\alpha_{jt}$ are independent across $j$.\footnote{See Footnote \ref{footnote_alpha} for the consequences of relaxing this assumption.}  We do not need to impose assumptions on $\theta_j$ and $\gamma_t$. 

{Since we do not restrict  the dependence between  $\mu_j$ and $D_j$, this sampling scheme can encompass settings in which we have relevant spatial correlation in the treatment assignment mechanism. For example, consider the case in which $\mu_j$ represents exposure to specific industry shocks. In this case, we may have that the probability of being assigned to treatment is larger for units that, for example, are more exposed to a specific industry, generating relevant spatial correlation. Likewise, if we think about the factors as representing geographical locations, then this formulation would allow for spatial correlation due to geographical distance.   }

The TWFE estimand in this case is given by  $\alpha \equiv \mathbb{E}[\frac{1}{T - t^\ast} \sum_{t \in \mathcal{T}_1} \alpha_{jt} | D_j =1]$, which we can think of as the population average treatment effects on the treated.  
If we let  $\mu^e = \mathbb{E}[ \mu_j]$,  and $\mu_w^e = \mathbb{E}[\mu_j | D_j=w]$, for $w \in \{ 0,1\}$, then 
\begin{eqnarray}  \label{alpha2}
\hat \alpha - \alpha = \frac{1}{N_1} \sum_{j \in \mathcal{I}_1} \left[ (\nabla \alpha_j - \alpha)+ \nabla \lambda(\mu_j - \mu^e)+  \nabla \epsilon_j \right] -  \frac{1}{N_0} \sum_{j \in \mathcal{I}_0} \left[ \nabla \lambda(\mu_j - \mu^e)+ \nabla \epsilon_j \right],
\end{eqnarray}
where, with some abuse of notation, $\nabla \alpha_j$ is the post-treatment average of $\alpha_{jt}$ across $t$.   

We consider a setting in which the linear factor structure does not affect the counterfactual trends, so the DID model is not misspecified. We impose the following assumption, which implies a standard parallel trends assumption $\mathbb{E}[\nabla Y_j(0) | D_j = 1] = \mathbb{E}[\nabla Y_j(0) | D_j = 0]$.\footnote{This assumption is implied by the assumption of parallel trends for all periods. We can extend our results to consider alternative parallel trends assumptions  \citep{Marcus}.} 

\begin{ass}{(parallel trends)}
\normalfont \label{assumption_unbiasedness}
$\mathbb{E}[ \nabla \epsilon_j | {D}_j] = 0$ and $\mathbb{E}[\nabla \lambda] (\mu_1^e - \mu_0^e) = 0$.
\end{ass}
 
The first part of Assumption \ref{assumption_unbiasedness} states that idiosyncratic errors are uncorrelated with treatment assignment. The second part implies that factor structure does not affect the expected value of the DID estimator.   Note that $\mathbb{E}[\nabla \lambda] (\mu_1^e - \mu_0^e) = \sum_{f=1}^F \mathbb{E}[\nabla \lambda(f) ]  (\mu_1^e(f) - \mu_0^e(f))$, where $v(f)$ is the $f-$th coordinate of vector $v$. If we do not take into account knife-edge cases in which elements of this sum cancel out, Assumption \ref{assumption_unbiasedness} implies that, for each $f=1,...,F$, either one of two conditions hold. 
 First,  it may be that $\mathbb{E}[\bar \lambda_{\mbox{\tiny post}}(f)] = \mathbb{E}[ \bar \lambda_{\mbox{\tiny pre}}(f)]$, so  the first moment of the distribution of the common factor $f$ is stable in the pre- and post-treatment periods. In this case, even if treated and control units are differentially affected by this common factor, this would not generate bias on the DID estimator over the distribution of $\lambda_t(f)$. Alternatively, it may be that $ \mu^e_1(f) = \mu^e_0(f)$. In this case, even if the expected value of $\lambda_t(f)$ differs in the pre- and post-treatment periods,  this common factor does not systematically affect treated units differently relative to control units, so this would not generate bias for the DID estimator over the distribution of $\mu_j(f)$.  Since we also have $\mathbb{E}[\nabla \alpha_j | D_j=1] = \alpha$, Assumption \ref{assumption_unbiasedness} implies that $\hat \alpha$ is unbiased.

Overall, we can think that there are unit- and/or  time-invariant unobserved variables that may be arbitrarily correlated with treatment assignment, but the other common shocks are not correlated with treatment assignment once we condition on these fixed effects.

\subsubsection{Asymptotic distribution}

In order to derive the asymptotic distribution of the DID estimator in this setting, we consider  a local-to-0 approximation in which the variance of $\nabla \lambda (\mu^e_1 - \mu^e_0)$ drifts to zero. This way, we can rely on an asymptotic theory to approximate settings in which the ratio between the variance of the common shocks and the variance of the average of the idiosyncratic shocks assumes any value in $[0,\infty)$. Therefore, we can consider approximations to settings in which common shocks have negligible, moderate, or large relevance relative to the sampling variation.\footnote{If we do not consider a local asymptotics,  this ratio would diverge when $N \rightarrow \infty$, and this would not provide reasonable approximations to many relevant applications. \cite{Roth} considers a similar assumption. We consider in Appendix \ref{Appendix_no_drift} the case in which $var(\nabla \lambda (\mu^e_1 - \mu^e_0)$ does not drift to zero.}

\begin{ass}{(local-to-0 approximation)}
\normalfont   \label{local0}
 $ \sqrt{N}  \lambda_t = \xi_t$, where $\mathbb{E}[\nabla \xi (\mu_1^e - \mu_0^e)] = 0$ and $var(\nabla \xi(\mu_1^e - \mu_0^e)) =(\mu_1^e - \mu_0^e)' \Omega(\mu_1^e - \mu_0^e)$.

\end{ass}

\begin{prop}
 \label{Proposition_local0}
 Consider a setting in which potential outcomes follow equation (\ref{LFM}), and treatment starts after periods $t^\ast$. Assumptions \ref{rv} to \ref{local0} hold. Then, as  $N \rightarrow \infty$,
 \begin{eqnarray}
\sqrt{N}(\hat \alpha -  \alpha)   \buildrel d \over \rightarrow \nabla \xi (\mu_1^e - \mu_0^e) + \frac{1}{c} \sigma_\epsilon(1)Z_1 + \frac{1}{1-c}\sigma_\epsilon(0) Z_0,  
\end{eqnarray}
where $Z_1$ and $Z_0$ are standard normal variables,  and $\nabla \xi$, $Z_0$ and $Z_1$ are mutually independent. For $w \in \{0,1\}$,  $\sigma^2_\epsilon(w) = var( \nabla \epsilon_j + (\nabla \alpha_j - \alpha)D_j | D_j = w)$.
Moreover, if $\alpha= 0$, then the $t$-statistic using CRVE at the unit level will be such that 
 \begin{eqnarray}
t = \frac{\hat \alpha}{\sqrt{\widehat{var(\hat \alpha)}_{\tiny \mbox{Cluster}} }}   \buildrel d  \over \rightarrow Z + V ,  
\end{eqnarray}
where $Z \sim N(0,1)$, $V = \frac{\nabla \xi (\mu_1^e - \mu_0^e)}{\sqrt{\frac{1}{c} \sigma^2_\epsilon(1) + \frac{1}{1-c} \sigma^2_\epsilon(0)}}$,  and $Z \perp V$.

\end{prop}

We present  details of the proof  in Appendix \ref{Proof_local0_appendix}. While $\hat \alpha$ is unbiased despite the spatial correlation, Proposition \ref{Proposition_local0} shows that $\hat \alpha$  may not be asymptotically normal if  $\nabla \xi (\mu^e_1 - \mu^e_0)$ is not normally distributed. As a consequence, we may have distortions for inference based on a $t$-statistic for two reasons. First, the asymptotic distribution of the  $t$-statistic, under the null, will have a variance greater than one. Second, the asymptotic distribution of the  $t$-statistic, under the null, may not be normal.

 If we assume that  $\nabla \xi(\mu_1^e - \mu_0^e)$ is normally distributed, then the $t$-statistic based on CRVE, under the null, would be asymptotically normal with mean zero, but its variance would be greater than one if  by $\Lambda_\lambda \equiv (\mu^e_1 - \mu^e_0)' \Omega (\mu^e_1 - \mu^e_0)>0$.

\begin{cor}
 \label{Corollary_local}
 
Consider the setting from Proposition \ref{Proposition_local0}, and assume further that $\nabla \xi (\mu_1^e - \mu_0^e) \sim N(0,(\mu_1^e - \mu_0^e)' \Omega (\mu_1^e - \mu_0^e))$. Then, if $\alpha = 0$,
\begin{eqnarray} \label{t_stat}
t = \frac{\hat \alpha}{\sqrt{\widehat{var(\hat \alpha)}_{\tiny \mbox{Cluster}} }}   \buildrel d \over \rightarrow N \left(0, 1 + \frac{(\mu_1^e - \mu_0^e)' \Omega (\mu_1^e - \mu_0^e)}{\frac{1}{c} \sigma_\epsilon^2(1) + \frac{1}{1-c} \sigma_\epsilon^2(0)} \right) \mbox{,   as  $N \rightarrow \infty$.}  
\end{eqnarray}
\end{cor}

\subsubsection{Size distortion when spatial correlation is ignored} \label{Distortions}

Proposition \ref{Proposition_local0} and Corollary \ref{Corollary_local} make it clear that, when $\Lambda_\lambda >0$, ignoring spatial correlation in this setting leads to over-rejection for two-sided $t$-tests.\footnote{Even if $V$ is not normal, $Z \sim N(0,1)$ and $Z \perp V$ implies that, for any $c \in \mathbb{R}$, $Pr(|Z+V|>c) = \int Pr(|Z+v|>c)dF_v(v) >  \int Pr(|Z|>c)dF_v(v) = Pr(|Z|>c)$. This inequality follows from $Pr(|Z + v | > c) > Pr(|Z| > c)$ for any $v \neq 0$, and $Pr(V = 0) < 1$.} Moreover,  over-rejection will be larger when $\Lambda_\lambda$ is larger relative to $\Lambda_\epsilon \equiv \frac{1}{c} \sigma_\epsilon^2(1) + \frac{1}{1-c} \sigma_\epsilon^2(0)$.\footnote{Under the assumption that $\nabla \alpha_j$ is iid, a larger treatment effect heterogeneity ($var(\nabla \alpha_j | D_j =1)$) leads to larger $\sigma_\epsilon^2(1)$, which in turn implies smaller underestimations by the CRVE. This happens because the treatment effect heterogeneity in this case is captured by the CRVE. However, we should expect the opposite in case there is strong spatial correlation in $\nabla \alpha_j$. Overall, whether larger treatment effects heterogeneity leads to more or less underestimation by the CRVE depends on the degree of spatial correlation in the treatment effects heterogeneity.  \label{footnote_alpha}    } Importantly, implementation details, such as the time frame used in the estimation and the choice of the control group will affect the relative magnitude between $\Lambda_\lambda$ and $\Lambda_\epsilon$. 

\noindent
\textbf{Choice of the time frame:} if the spatially correlated shocks are also more serially correlated than the idiosyncratic shocks, then considering shorter time frames around the treatment would lead to less size distortions. More specifically, assume that $\xi_t (\mu_1^e - \mu_0^e)$ follows an AR(1) process with serial correlation $\rho_\xi$, while $\epsilon_{jt}$, conditional on either $D_j = 0$ or $D_j = 1$, follows an AR(1) process with serial correlation $\rho_\epsilon$. Consider a DID estimator using $T/2$ periods before and $T/2$ periods after the treatment. As we show in  Appendix  \ref{appendix_dependence}, if $0 \leq  \rho_\epsilon < \rho_\xi < 1$, then inference distortions based on CRVE are increasing in $T$.\footnote{In Appendix \ref{appendix_dependence}, we derive the formula for $\phi(\rho_\xi,\rho_\epsilon,T) = \Lambda_\lambda / \Lambda_\epsilon$, in a setting in which $\xi_t$ and $\epsilon_{j,t}$ are AR(1), and we have a DID estimator with $T$ periods. We show numerically that $\phi(\rho_\xi,\rho_\epsilon,T)$ is increasing in $T$ when $ 0 \leq \rho_\epsilon < \rho_\xi <1 $ for all reasonable values of  $T$.   }  

An important caveat is that considering different time frames  implies that the DID estimand may change. More specifically, the estimand would be given by $\mathbb{E}[\frac{1}{\tilde t} \sum_{t \in \widetilde{\mathcal{T}}} \alpha_{jt} | D_j =1]$, where $\widetilde{\mathcal{T}}$ is the set of  post-treatment time periods used to construct the DID estimator. Therefore, if we consider shorter time frames, then we would only estimate  short-term effects of the policy. Note also that changing the time frame also implies that we would consider a modified Assumption \ref{assumption_unbiasedness}, which can be valid considering only the time periods used for the estimation.

Likewise, consider a dynamic DID, in which one uses a base period (for example, $t^\ast$) and a period $t^\ast + \tau$ for varying $\tau$, where $\tau < 0$ provides evidence on the parallel trends assumptions, while $\tau > 0$ provides estimates for the effect  $\tau$ periods after the treatment. Our results also imply  that the degree of size distortions when spatial correlation is ignored can vary substantially  for different values of $\tau$. In particular, under the assumption $0 \leq  \rho_\epsilon < \rho_\xi < 1$,  we should expect more size distortions when $|\tau|$ increases (more details in Appendix \ref{appendix_dependence}).  

\noindent
\textbf{Choice of treated and control units:}  Proposition \ref{Proposition_local0}  and Corollary \ref{Corollary_local} also imply that size distortions would be lower when $\mu_1^e \approx \mu_0^e$. In such cases, the time fixed effects would absorb most of the spatial correlation, and inference based on CRVE at the unit level would lead to less distortions.\footnote{This is related to the idea of using state-border DID, as considered by \cite{Dube}. }

In settings in which the nature of the spatial correlation is unknown,  it would not be possible to select the treated/control group taking that into account. However, in some settings this conclusion may be useful in practice. For example, consider a setting in which we observe students from grades one to four, and consider a treatment that starts in the post-treatment periods for grades three and four. 
 In this case, if closer grades are more similarly exposed to the common shocks relative to more distant grades, then  we should expect smaller size distortions if we consider a DID estimator comparing students from grades two and three, relative to a DID estimator using the full sample. In Section \ref{Sim_CPS},  we present simulations based on the CPS that corroborate this conclusion. An important caveat is that, if treatment effects are heterogenous, then this approach might  change the DID estimand. In this example, we would estimate the effects for third graders (instead of an average effect for third and fourth graders).

\begin{remark}
\normalfont

We consider in Appendix \ref{Appendix_other} (i) the case in which the variance of $\lambda_t$ does not drift to zero, (ii) a setting with $F \rightarrow \infty$, and (iii)  a design-based approach for inference. Our main conclusions remain valid for these settings. 
\end{remark}

\begin{remark}
\normalfont

\label{Remark_estimators}
Other estimators, such as the first-difference  or the recent set of estimators proposed for settings with variation in treatment timing can generally be seen as combinations of simpler 2x2 DID estimators.\footnote{For example, \cite{chaisemartin2018twoway}, \cite{clement2},  \cite{Pedro}, and \cite{SUN2020}} Therefore, our results also apply to these other estimators. In particular, some of these recently proposed estimators focus on short time-differences. As a consequence, if the spatially correlated shocks are also more serially correlated than the idiosyncratic shocks, such estimators have the additional benefit of being less affected by spatial correlation. 

\end{remark}

\begin{remark}
\normalfont  \label{Remark_pretest}

We  show in Appendix \ref{pretest} that usual pre-tests for parallel trends can also capture inference problems due to spatial correlation, in addition to providing evidence on departures from parallel trends. Differently from the results by  \cite{Roth}, who shows that pre-testing may exacerbate the problem of violations of parallel trends, we show that pre-testing does not exacerbate the inference problem if the only problem is spatial correlation.

\end{remark}

\section{Monte Carlo Simulations} \label{simulations}

We consider two sets of simulations, one that mimics a setting in which the relevant source of spatial correlation is unknown by the econometrician, and another one in which it is known, but there is not enough variation to take that into account. 

\subsection{Unknown (by the econometrician) spatial correlation} \label{Sim_ACS}

We  first consider  MC simulations in which we estimate a spatial correlation structure based on the ACS \citep{ipums}. We aggregate the data at the Public Use Microdata Area (PUMA) $\times$ year level, considering 2005 to 2019.   Following \cite{Bertrand04howmuch}, we restrict the sample to women between the ages 25 and 50, and focus on log wages as the outcome variable. We consider simulations in which we fix the total number of years, $T$, and treatment starts in the middle of the time frame. For a given $T \in \{2,3,...,15 \}$, we estimate a covariance matrix in which  $cov(W_j,W_k)$ may depend on whether PUMA's $j$ and $k$ are in the same state, and/or  on whether they  have  similar industry compositions. We present  details on how the DGP is  constructed in Appendix \ref{A_ACS}. 

For a given  $T$, we simulate a Gaussian model with the estimated covariance structure for such  $T$.\footnote{Note that the covariance structure estimated from the ACS is identified even when errors are not normal. In particular, if the estimated covariance structure is such that $var(Z + V) \approx 1$, then $var(V) \approx 0$ (where $Z$ and $V$ are defined in Proposition \ref{Proposition_local0}). Therefore, we should expect the same patterns regarding settings in which spatial correlation does not lead to large distortions if we did not assume a Gaussian DGP. Consistent with that, a previous version of this paper \citep{Ferman_OLD} presented simulations in a design-based approach, in which we did not impose that the potential outcomes are normal, and found similar results. Moreover, the simulations in Section \ref{Sim_CPS} do not assume normality, and find similar results. }  We consider two alternative treatment assignment mechanisms. In the first one, we consider PUMA's completely randomly assigned. As presented in Figure \ref{fig_ACS}.A, even if we consider CRVE at the PUMA level,  rejection rates are close to 5\% regardless of the spatial correlation in the DGP.   This is consistent with the conclusions from \cite{IK}. 

In the second assignment mechanism, we consider a setting in which PUMA's with industry compositions that are more concentrated in manufacturing have higher probability of receiving treatment. Therefore, in this case  $\mu_1^e \neq \mu_0^e$ for the common shocks related to industry composition. Still, since $\mathbb{E}[W_j | D_j]=0$ for all $j$, the DID estimator is unbiased.\footnote{We provide evidence that this is a reasonable assumption in this setting in  Appendix \ref{A_ACS}. }

  It is conceivable that an applied researcher might be unaware about (or may not have information on) such industry-level shocks. Therefore, we consider  first  inference based on CRVE at the PUMA level. In this case,  rejection rates are relatively close to 5\% when $T$ is small (for example, at 7\% when $T=2$), but over-rejection becomes more problematic  when $T$  increases, reaching  19\% when $T=15$ (Figure \ref{fig_ACS}.B). Considering the results from Section \ref{Distortions}, this is consistent with industry-level shocks being more serially correlated relative to the idiosyncratic shocks.\footnote{We recall that the parameters of the spatial correlation in this DGP were estimated based on a real (and widely used by applied researchers) dataset. If it were the case that the data is such that the idiosyncratic shocks are relatively more serially correlated, then we should expect the reverse pattern in terms of size distortions in Figure  \ref{fig_ACS}.B. } More generally, this example illustrates that the relevance of (ignored) spatial correlation  depends crucially on the  time frame considered  in the application. We present in Appendix Figure \ref{App_fig_dynamic} results for a dynamic DID specification. We similarly find that size distortions are minor when we consider estimation of shorter-term effects, but become more relevant when we consider longer-term effects.

Now consider that the applied researcher attempts to correct for spatial correlation, but considers a geographical distance as the relevant distance metric. We consider  a wild cluster bootstrap (WCB) at the state level.\footnote{Results with CRVE at the state level are similar, but with slightly larger rejection rates due to some large state clusters.} Over-rejection becomes slightly smaller in most cases, but we still find relevant over-rejection when   $T $ is large. The reason is that the state-level cluster captures some of the industry-level shocks, because some states have a relatively higher concentration of PUMA's more exposed to manufacturing. However, we still have relevant over-rejection when $T$ is large, because there are relevant across-state correlations that are not taken into account. We also consider a border DID approach. Again, this approach slightly  ameliorates the inference problem, but does not completely solve it. 

Finally, not surprisingly, if the applied researcher had complete knowledge that the relevant spatial correlation came from such industry shocks, then clustering at the industry-group level would be valid regardless of $T$.

\subsection{Known spatial correlation} \label{Sim_CPS}

We now present simulations using the CPS data from 1990 to 2018, still considering log wages for women between the ages of 25 and 50. For each simulation, in addition to selecting a time frame with $T \in \{2,3,...,10 \}$, we also select an age frame with $\delta_{\mbox{\tiny age}} \in \{2,3,...,10 \}$. We construct a DGP based on this dataset in which we allow for individuals of similar ages to be more spatially correlated, in addition to allowing for within-state correlation and for serial correlation. We present in details how this DGP is constructed in  Appendix \ref{A_CPS}. 
Given $T$ and $\delta_{\mbox{\tiny age}}$, we consider simulations in which treatment starts in the second half of the years for individuals above median in terms of age. In this setting, we expect relevant spatial correlation if individuals of closer ages  (whether or not they are in the same state) are likely to be affected by similar shocks. We consider DID regressions including state $\times$ age-group fixed effects,  time fixed effects, and the DID dummy.  In those simulations, the DID estimator is unbiased, and the null hypothesis is true.\footnote{We present evidence in Appendix  \ref{A_CPS} that it is reasonable to assume  parallel trends in these simulations.}

Table \ref{Table_CPS} presents rejection rates when inference is based on CRVE at the state level. There is large over-rejection (with rejection rates up to 36\%) when \emph{both} the time and the age frames are large. In contrast, there is not much over-rejection when $T$ is small, regardless of the age frame. This is again consistent with spatially correlated shocks being relatively more serially correlated relative to the idiosyncratic shocks. More interesting, even when $T$ is large, there is not much over-rejection when we keep the age frame small. This is consistent with the theoretical results that the distortions are mitigated when treated and control groups are more similar. Differently from the setting considered in Section \ref{Sim_ACS}, in this setting it would be possible to select treated and control groups in such a way. 

In Appendix \ref{Appendix_alternatives_CPS}, we show that spatial-correlation robust standard errors do not work well in these simulations, given that we  have little variation in the age groups, even when $\delta_{\mbox{\tiny age}}=10$.

\section{Recommendations \& Concluding Remarks} 

\label{recommendations}

Spatial correlation can lead to substantial over-rejection. Whenever feasible,  applied researchers should consider methods that take that into account, as the ones discussed in Section \ref{Existing_methods}. However,  there are common settings in which such solutions are unfeasible. Also, even when feasible, some alternatives may involve relevant trade-offs. For example, with variation in treatment timing, the  CCE estimator may take spatial correlation into account, but there are some limitations and trade-offs in considering this alternative (see details in Appendix \ref{Appendix_LFM}). Therefore, even if an applied researcher decides to use the CCE estimator, it might be valuable to also consider DID alternatives as a robustness check.\footnote{For example, it might be interesting to consider a DID estimator that deals with the problem of aggregating heterogeneous treatment effects when there is variation in treatment timing, which is a potential problem for the CCE estimator.} 

Given that, the results we present in this paper provide guidelines on how applied researchers could proceed in empirical applications to mitigate and assess the relevance of  spatial correlation when relying on inference methods that assume independent errors in the cross-section (such as CRVE at the unit level).

Consider a setting with more than one pre- and post-treatment periods. In this case, a longer time series would imply larger over-rejection if common factors exhibit stronger  serial correlation relative to the idiosyncratic shocks. The simulations from Section \ref{simulations} provide evidence that this is the case for the ACS and CPS datasets.  One robustness check in this case is to consider a specification  restricting the sample to a few periods before and a few periods after the treatment. In this case, the unit fixed effects would absorb more of these common shocks, making inference assuming independent units more reliable.  An important caveat is that, in this case, the DID estimand would provide the short-term effect of the policy.   As discussed in Remark \ref{Remark_pretest}, in this setting it would also possible to use pre-treatment data to check  whether inference based on short differences is indeed reliable.

Relatedly, we also show that, in dynamic DID specifications, size distortions may vary substantially, depending on the time horizon that we analyze. In particular, if  common factors exhibit stronger  serial correlation relative to the idiosyncratic shocks, then we should expect relatively larger inference distortions for longer-term effects.

Another alternative to mitigate the spatial correlation problem is to make  treated and control units  as similar as possible. This alternative is unfeasible if the source of spatial correlation is unknown. However, as  illustrated in the simulations in Section \ref{Sim_CPS}, this can be a valid alternative in case there is information about the source of spatial correlation, but spatial correlation-robust standard errors do not work well. 

Overall, this paper analyzes the challenges for inference in DID when there is spatial correlation. We present a series of novel insights and empirical evidence on the settings in which ignoring spatial correlation should lead to more or less distortions in DID applications. We show that details such as the time frame used in the estimation, the choice of the treated and control groups, and the choice of the estimator, are key determinants of  distortions due to spatial correlation.   We also analyze in detail the feasibility and trade-offs involved in a series of  alternatives to take spatial correlation into account. Given that, we provide relevant recommendations for applied researchers on how to mitigate and assess the possibility of inference distortions due to spatial correlation.



\singlespace

\renewcommand{\refname}{References} 

\bibliographystyle{apalike}
\bibliography{bib/bib.bib}

\pagebreak

\begin{figure}[H] 

\begin{center}
\caption{{\bf Simulations with the ACS}}  \label{fig_ACS}

\begin{tabular}{cc}

A. Random assignment & B. Non-random assignment \\

\includegraphics[scale=0.6]{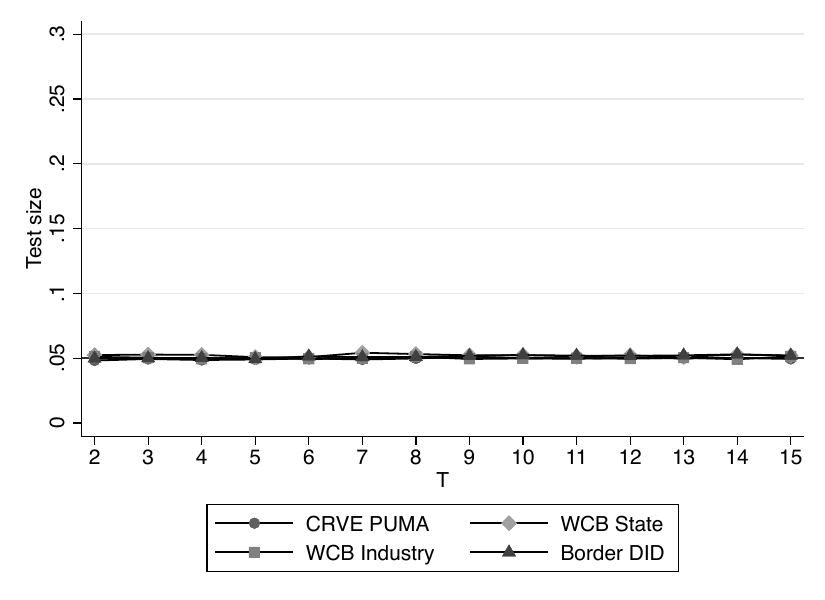} & \includegraphics[scale=0.6]{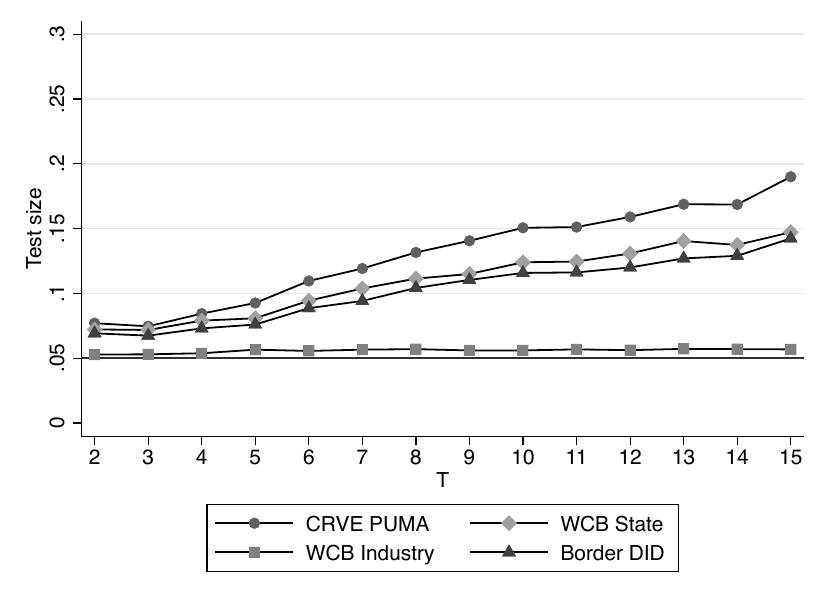}

\end{tabular}

\end{center}

\small{Notes: This figure presents rejection rates for the simulations based on the ACS data, as a function of the time frame used in the estimation. In each simulation, we run a DID regression and we consider inference based on CRVE at the PUMA level, WCB at the state level, and WCB at the industry level. We also consider a border DID regression. Details on the construction of these simulations are presented in  Section \ref{Sim_ACS} and Appendix \ref{A_ACS}.  }

\end{figure}

\pagebreak

\begin{table}[H]
  
 \begin{center}
\caption{{\bf Simulations with the CPS}} \label{Table_CPS}

\begin{tabular}{ccccccccccc}
\hline
\hline

& & \multicolumn{9}{c}{Time frame} \\ \cline{3-11}

\multirow{9}{*}{\rotatebox{90}{Age frame}}  & & 2 & 3 & 4 & 5 & 6 & 7 & 8 & 9 & 10 \\

 & 2 & 0.044 & 0.055 & 0.058 & 0.061 & 0.063 & 0.061 & 0.062 & 0.067 & 0.072 \\
 & 3 & 0.054 & 0.054 & 0.061 & 0.058 & 0.069 & 0.058 & 0.061 & 0.069 & 0.063 \\
 & 4 & 0.048 & 0.050 & 0.064 & 0.070 & 0.078 & 0.078 & 0.077 & 0.077 & 0.085 \\
 & 5 & 0.049 & 0.062 & 0.059 & 0.062 & 0.087 & 0.105 & 0.097 & 0.117 & 0.116 \\
 & 6 & 0.048 & 0.062 & 0.059 & 0.069 & 0.099 & 0.122 & 0.123 & 0.136 & 0.164 \\
 & 7 & 0.048 & 0.061 & 0.077 & 0.092 & 0.113 & 0.143 & 0.168 & 0.188 & 0.220 \\
 & 8 & 0.041 & 0.055 & 0.083 & 0.112 & 0.129 & 0.156 & 0.203 & 0.221 & 0.284 \\
 & 9 & 0.048 & 0.058 & 0.085 & 0.118 & 0.153 & 0.191 & 0.237 & 0.278 & 0.331 \\
 & 10 & 0.048 & 0.063 & 0.097 & 0.120 & 0.184 & 0.233 & 0.267 & 0.333 & 0.365 \\

\hline

\end{tabular}

 \end{center}
{\small{Notes:   This table presents rejection rates for the simulations using CPS data, as a function of the time frame and age frame used in the simulations. Details  presented in Section \ref{Sim_CPS}.   For each simulation, we run a DID regression and test the null hypothesis using CRVE at the state level.  Details on the construction of these simulations are presented in  Section \ref{Sim_CPS} and Appendix \ref{A_CPS}.    }}

\end{table}

\pagebreak

\appendix

 \setcounter{table}{0}
\renewcommand\thetable{A.\arabic{table}}

\setcounter{figure}{0}
\renewcommand\thefigure{A.\arabic{figure}}
 
\section{Appendix (for online publication)}

\onehalfspacing

\subsection{Proof of the main results}

\subsubsection{Proof of Proposition \ref{Proposition_local0}} \label{Proof_local0_appendix}

\begin{proof}

From equation (\ref{alpha2}),
\begin{eqnarray}  \nonumber
\sqrt{N}(\hat \alpha - \alpha) &=& \sqrt{N} (\nabla \lambda)(\mu_1^e - \mu_0^e) +  \sqrt{N} (\nabla \lambda) \frac{1}{N_1} \sum_{j \in \mathcal{I}_1} (\mu_j - \mu_1^e)+  \sqrt{N} \frac{1}{N_1} \sum_{j \in \mathcal{I}_1} [ (\nabla \alpha_j - \alpha)+ \nabla \epsilon_j ] \\
&& -  \sqrt{N} (\nabla \lambda) \frac{1}{N_0} \sum_{j \in \mathcal{I}_0} (\mu_j - \mu_0^e)-  \sqrt{N} \frac{1}{N_0} \sum_{j \in \mathcal{I}_0} \nabla \epsilon_j .
\end{eqnarray}

Note that $ \sqrt{N} (\nabla \lambda) = \nabla \xi =  O_p(1)$, ${N_w}^{-1} \sum_{j \in \mathcal{I}_w} (\mu_j - \mu_w^e) = o_p(1)$, and ${N_w}^{-1/2} \sum_{j \in \mathcal{I}_w}  [(\nabla \alpha_j - \alpha)D_j + \nabla \epsilon_j ] \buildrel d \over \rightarrow  N(0,\sigma_\epsilon^2(w))$. Moreover,  $\sqrt{N} (\nabla \lambda)$, ${N_0}^{-1/2} \sum_{j \in \mathcal{I}_0} \nabla \epsilon_j  $ and ${N_1}^{-1/2} \sum_{j \in \mathcal{I}_1}  [(\nabla \alpha_j - \alpha) + \nabla \epsilon_j ] $ are mutually independent. Therefore, 
\begin{eqnarray}
\sqrt{N}(\hat \alpha -  \alpha)   \buildrel d \over \rightarrow \nabla \xi (\mu_1^e - \mu_0^e) + \frac{1}{c} \sigma_\epsilon(1)Z_1 - \frac{1}{1-c}\sigma_\epsilon(0) Z_0,  
\end{eqnarray}
where $Z_1$ and $Z_0$ are standard normal variables, and $\xi$, $Z_1$ and $Z_2$ are mutually independent. 

Moreover, the OLS residuals from TWFE DID regression are such that, for $j \in \mathcal{I}_w$, $w \in \{0,1\}$,  
\begin{eqnarray} 
\widehat{W}_j &=& \nabla Y_{j} - \frac{1}{N_w} \sum_{k \in \mathcal{I}_w} \nabla Y_{j} \\ \nonumber
&=&\nabla \lambda (\mu_j - \mu^e_w) + (\nabla \alpha_{jt} - \alpha)D_j +  \nabla \epsilon_{j} -  \frac{1}{N_w}\sum_{k \in \mathcal{I}_w} \left[ \nabla \lambda (\mu_k - \mu^e_w) +  (\nabla \alpha_{kt} - \alpha)D_k + \nabla \epsilon_{k} \right].
\end{eqnarray}

For $w \in \{0,1\}$, let   $\Sigma_\mu(w) = var(\mu_j | D_j = w)$.  Given Assumptions \ref{rv} and \ref{assumption_unbiasedness}, 
\begin{eqnarray} \label{error}
\frac{1}{N_w} \sum_{j \in \mathcal{I}_w} \widehat{W}_j^2  = (\nabla \lambda)( \Sigma_\mu(w))( \nabla \lambda' ) + \sigma_\epsilon^2(w) + o_p(1).
\end{eqnarray}  

Given Assumption \ref{local0},   ${N_w}^{-1} \sum_{j \in \mathcal{I}_w} \widehat{W}_j^2  = \sigma_\epsilon^2(w) + o_p(1)$. Therefore,

\begin{eqnarray*} 
N \widehat{var(\hat \alpha)}_{\tiny \mbox{Cluster}} &=& \frac{N}{N_1} \left( \frac{1}{N_1} \sum_{j \in \mathcal{I}_1} \widehat{W}_j^2 \right) + \frac{N}{N_0} \left( \frac{1}{N_0} \sum_{j \in \mathcal{I}_0} \widehat{W}_j^2 \right)   \\
&=&  \frac{1}{c}\sigma_\epsilon^2(1) +  \frac{1}{1-c}\sigma_\epsilon^2(0)  +o_p(1).
\end{eqnarray*}  
\end{proof}

\subsection{Different settings} \label{Appendix_other}

\subsubsection{Case in which the variance of $\lambda_t$ does not drift to zero} \label{Appendix_no_drift}

We consider the case in which the variance of $\lambda_t$ does not drift to zero. In this case, we have the following proposition. 

\begin{prop}
 \label{Proposition}
Consider a setting in which potential outcomes follow equation (\ref{LFM}), and treatment starts after periods $t^\ast$. Assumptions \ref{rv} and \ref{assumption_unbiasedness} hold. Then, as  $N \rightarrow \infty$,
\begin{eqnarray}
\hat \alpha -  \alpha = \nabla \lambda (\mu_1^e - \mu_0^e) + o_p(1),  
\end{eqnarray}
\begin{eqnarray}
var(\hat \alpha) - \widehat{var(\hat \alpha)}_{\tiny \mbox{Cluster}}  = (\mu^e_1 - \mu^e_0)' \mathbb{E} \left[ \left( \nabla \lambda \right)'\left( \nabla \lambda \right)  \right] (\mu^e_1 - \mu^e_0) + o_p(1),
\end{eqnarray}
and
\begin{eqnarray}
 \widehat{var(\hat \alpha)}_{\tiny \mbox{Cluster}}  =  o_p(1).
\end{eqnarray}
\end{prop}

\begin{proof}

Note first that 
\begin{eqnarray}  \label{alpha2_proof} \nonumber
\hat \alpha - \alpha &=& \nabla \lambda(\mu_1^e - \mu_0^e) +  \frac{1}{N_1} \sum_{j \in \mathcal{I}_1} \left[ \nabla \lambda(\mu_j - \mu_1^e)+ (\nabla \alpha_j - \alpha)+ \nabla \epsilon_j \right] -  \frac{1}{N_0} \sum_{j \in \mathcal{I}_0} \left[ \nabla \lambda(\mu_j - \mu_0^e)+ \nabla \epsilon_j \right] \\
&=&  \nabla \lambda(\mu_1^e - \mu_0^e) + o_p(1),
\end{eqnarray}
since the terms $ \nabla \lambda(\mu_j - \mu_w^e)$, $(\nabla \alpha_j - \alpha)$, and $ \nabla \epsilon_j$ are uncorrelated across $j$. 

The OLS residuals from TWFE DID regression are such that, for $j \in \mathcal{I}_w$, $w \in \{0,1\}$,  
\begin{eqnarray} 
\widehat{W}_j &=& \nabla Y_{j} - \frac{1}{N_w} \sum_{k \in \mathcal{I}_w} \nabla Y_{j} \\ \nonumber
&=&\nabla \lambda (\mu_j - \mu^e_w) + (\nabla \alpha_{jt} - \alpha)D_j +  \nabla \epsilon_{j} -  \frac{1}{N_w}\sum_{k \in \mathcal{I}_w} \left[ \nabla \lambda (\mu_k - \mu^e_w) +  (\nabla \alpha_{kt} - \alpha)D_k + \nabla \epsilon_{k} \right].
\end{eqnarray}

For $w \in \{0,1\}$, let   $\Sigma_\mu(w) = var(\mu_j | D_j = w)$.  Given Assumptions \ref{rv} and \ref{assumption_unbiasedness}, 
\begin{eqnarray} \label{error}
\frac{1}{N_w} \sum_{j \in \mathcal{I}_w} \widehat{W}_j^2  = (\nabla \lambda)( \Sigma_\mu(w))( \nabla \lambda' ) + \sigma_\epsilon^2(w) + o_p(1).
\end{eqnarray}  

Therefore, 
\begin{eqnarray*} 
 \widehat{var(\hat \alpha)}_{\tiny \mbox{Cluster}} &=& \frac{1}{N_1} \left( \frac{1}{N_1} \sum_{j \in \mathcal{I}_1} \widehat{W}_j^2 \right) + \frac{1}{N_0} \left( \frac{1}{N_0} \sum_{j \in \mathcal{I}_0} \widehat{W}_j^2 \right)   \\
&=& \frac{1}{N_1} (\nabla \lambda)( \Sigma_\mu(1))( \nabla \lambda' ) + \frac{1}{N_1}\sigma_\epsilon^2(1) + \frac{1}{N_0} (\nabla \lambda)( \Sigma_\mu(0))( \nabla \lambda' ) \\
&& + \frac{1}{N_0}\sigma_\epsilon^2(0)  + o_p(N^{-1}) = o_p(1).
\end{eqnarray*}  

Now note that, under Assumptions \ref{rv} and \ref{assumption_unbiasedness}, 
\begin{eqnarray*} 
var\left(\hat \alpha | \mathbf{D}=\mathbf{d} \right) &=&  (\mu^e_1 - \mu^e_0)' \mathbb{E} \left[ \left( \nabla \lambda \right)'\left( \nabla \lambda \right)  \right] (\mu^e_1 - \mu^e_0)  + \frac{1}{N_1} \sigma^2_\epsilon(1) + \frac{1}{N_0} \sigma^2_\epsilon(0) \\
&& + \frac{1}{N_1} \mathbb{E}\left[ (\nabla \lambda)( \Sigma_\mu(1) )(\nabla \lambda)' \right] + \frac{1}{N_0} \mathbb{E}\left[ (\nabla \lambda)( \Sigma_\mu(0) )(\nabla \lambda)' \right],
\end{eqnarray*}  
where we implicitly assume that we are conditioning on $N_1 \geq 1$ and $N_0 \geq 1$. Otherwise, it would not be possible to construct a DID estimator.

Therefore,
\begin{eqnarray*} 
var\left(\hat \alpha  \right) &=& \mathbb{E}[var\left(\hat \alpha | \mathbf{D} \right) ] +  var[\mathbb{E}\left(\hat \alpha | \mathbf{D} \right) ]    \\
&=&  (\mu^e_1 - \mu^e_0)' \mathbb{E} \left[ \left( \nabla \lambda \right)'\left( \nabla \lambda \right)  \right] (\mu^e_1 - \mu^e_0)  + \mathbb{E} \left[ \frac{1}{N_1}\right] \sigma^2_\epsilon(1) +\mathbb{E} \left[ \frac{1}{N_0}\right] \sigma^2_\epsilon(0) \\
&& +\mathbb{E} \left[ \frac{1}{N_1}\right] \mathbb{E}\left[ (\nabla \lambda)( \Sigma_\mu(1) )(\nabla \lambda)' \right]  +\mathbb{E} \left[ \frac{1}{N_0}\right]  \mathbb{E}\left[ (\nabla \lambda)( \Sigma_\mu(0) )(\nabla \lambda)' \right] \\
&=&   (\mu^e_1 - \mu^e_0)' \mathbb{E} \left[ \left( \nabla \lambda \right)'\left( \nabla \lambda \right)  \right] (\mu^e_1 - \mu^e_0) + o(1),
\end{eqnarray*}  
since $\mathbb{E}[N_w^{-1}] = o(1)$ from  $N_w^{-1} \buildrel p \over \rightarrow  0$ and $ | N_w^{-1} | \leq 1$, and $\mathbb{E}\left(\hat \alpha | \mathbf{D} \right) =\alpha$. 

Therefore, 
\begin{eqnarray}
var(\hat \alpha) - \widehat{var(\hat \alpha)}_{\tiny \mbox{Cluster}}  &=& (\mu^e_1 - \mu^e_0)' \mathbb{E} \left[ \left( \nabla \lambda \right)'\left( \nabla \lambda \right)  \right] (\mu^e_1 - \mu^e_0) + o_p(1). 
\end{eqnarray}
\end{proof}

Note that, in this case, the DID estimator is unbiased, but is not consistent if $(\mu^e_1 - \mu^e_0)' \mathbb{E} \left[ \left( \nabla \lambda \right)'\left( \nabla \lambda \right)  \right] (\mu^e_1 - \mu^e_0)>0$. The CRVE underestimates the true variance of the DID estimator by $(\mu^e_1 - \mu^e_0)' \mathbb{E} \left[ \left( \nabla \lambda \right)'\left( \nabla \lambda \right)  \right] (\mu^e_1 - \mu^e_0)$. Therefore, we have that variance will be less underestimated under exactly the same conditions as we find in Proposition \ref{Proposition_local0}. That is, when the second moments of $\nabla \lambda$ are close to zero, and/or $\mu^e_1 \approx \mu^e_0$.  

Since $ \widehat{var(\hat \alpha)}_{\tiny \mbox{Cluster}} = o_p(1)$, the variance of the t-statistic based on CRVE diverges if  $(\mu^e_1 - \mu^e_0)' \mathbb{E} \left[ \left( \nabla \lambda \right)'\left( \nabla \lambda \right)  \right] (\mu^e_1 - \mu^e_0)>0$  when $N \rightarrow \infty$. Therefore, even when the null is true, the probability of rejection would generally converge in probability to one.\footnote{This is true whenever $\nabla \lambda$ has a continuous distribution. If $\nabla \lambda$ had a probability mass at zero, we would not have the probability of rejection converging in probability to one.  Moreover, this is valid when the distribution of $\nabla \lambda$ is fixed when $N$ increases. We consider next a case in which variance of $\nabla \lambda$ goes to zero when $N\rightarrow \infty$.}

We consider now the case in which $(\mu^e_1 - \mu^e_0)' \mathbb{E} \left[ \left( \nabla \lambda \right)'\left( \nabla \lambda \right)  \right] (\mu^e_1 - \mu^e_0) = 0$, but $var(\nabla \lambda)$ does not drift to zero. From equation (\ref{alpha2_proof}),
\begin{eqnarray}  \nonumber
\sqrt{N}(\hat \alpha - \alpha) &=&  \nabla \lambda \frac{ \sqrt{N}}{N_1} \sum_{j \in \mathcal{I}_1} (\mu_j - \mu_1^e)+   \frac{ \sqrt{N}}{N_1} \sum_{j \in \mathcal{I}_1} [ (\nabla \alpha_j - \alpha)+ \nabla \epsilon_j ] \\
&& - \nabla \lambda \frac{ \sqrt{N}}{N_0} \sum_{j \in \mathcal{I}_0} (\mu_j - \mu_0^e) -   \frac{ \sqrt{N}}{N_0} \sum_{j \in \mathcal{I}_0} \nabla \epsilon_j .
\end{eqnarray}

Therefore, assuming that $\mu_j $ and $\nabla \epsilon_j$ are independent,  the asymptotic distribution of $\sqrt{N}(\hat \alpha - \alpha) $ is given by 
\begin{eqnarray*}  
\sqrt{N}(\hat \alpha - \alpha) &  \buildrel d \over \rightarrow&   \nabla \lambda \left( \frac{1}{\sqrt{c}} V_1 - \frac{1}{\sqrt{1-c}}  V_0 \right) +   \frac{1}{c} \sigma_\epsilon(1)Z_1 - \frac{1}{1-c}\sigma_\epsilon(0) Z_0,
\end{eqnarray*}
where $V_w \sim N(0, \Sigma_\mu(w))$, and $\nabla \lambda$, $V_1$, $V_0$, $Z_1$, and $Z_0$ are mutually independent.

The first conclusion is that the DID estimator is consistent, but it is generally not asymptotically normal. Note that the asymptotic distribution of $\hat \alpha$ is closer to normal if the second moments of $\nabla \lambda$ are closer to zero.

Moreover, we have that the asymptotic variance of $\hat \alpha$ is given by
\begin{eqnarray}  
a.var(\sqrt{N}(\hat \alpha - \alpha)) &=& \frac{1}{c} \sigma^2_\epsilon(1) + \frac{1}{1-c}\sigma^2_\epsilon(0)  + \frac{1}{c} \mathbb{E}\left[ (\nabla \lambda)( \Sigma_\mu(1) )(\nabla \lambda)' \right] \\
&&  +\frac{1}{1-c}\mathbb{E}\left[ (\nabla \lambda)( \Sigma_\mu(0) )(\nabla \lambda)' \right].
\end{eqnarray}

In contrast, we have that 
\begin{eqnarray} \nonumber
N \widehat{var(\hat \alpha)}_{\tiny \mbox{Cluster}} &=& \frac{1}{c} (\nabla \lambda)( \Sigma_\mu(1))( \nabla \lambda' ) + \frac{1}{c}\sigma_\epsilon^2(1) + \frac{1}{1-c} (\nabla \lambda)( \Sigma_\mu(0))( \nabla \lambda' ) + \frac{1}{1-c}\sigma_\epsilon^2(0) + o_p(1),
\end{eqnarray}  
implying that 
\begin{eqnarray} \nonumber
a.var(\sqrt{N}(\hat \alpha - \alpha)) - N \widehat{var(\hat \alpha)}_{\tiny \mbox{Cluster}}   &=&  \frac{1}{c} \left \{ \mathbb{E}\left[ (\nabla \lambda)( \Sigma_\mu(1) )(\nabla \lambda)' \right]  - (\nabla \lambda)( \Sigma_\mu(1))( \nabla \lambda' ) \right\}  \\ \nonumber
&& +  \frac{1}{1-c} \left \{ \mathbb{E}\left[ (\nabla \lambda)( \Sigma_\mu(0) )(\nabla \lambda)' \right]  - (\nabla \lambda)( \Sigma_\mu(0))( \nabla \lambda' ) \right\} \\ \nonumber
&& + o_p(1).
\end{eqnarray}

Therefore, another distortion comes from the fact that spatial correlation implies that we would not have a consistent estimator for the asymptotic variance of $\hat \alpha$, because the residuals would depend on the realization of $\nabla \lambda$. In this case, even asymptotically, the CRVE (multiplied by $N$) would differ from the asymptotic variance of $\hat \alpha$ due to the differences $ (\nabla \lambda)( \Sigma_\mu(w))( \nabla \lambda' ) - \mathbb{E}\left[ (\nabla \lambda)( \Sigma_\mu(w) )(\nabla \lambda)' \right]$ for $w \in \{0,1\}$. While the expected values of these differences are equal to zero, this can generate some size distortions, because the distribution of the test statistic would not be asymptotically normal. Again,  if $\lambda_t$ is serially positively correlated, with stronger dependence relative to the idiosyncratic shocks, then these terms become less relevant when we consider shorter time ranges. These terms also become less relevant if $\Sigma_\mu(w) = var(\mu_j | D_j = w) \approx 0$.

Finally, if we relax Assumption \ref{rv} to allow $\mu_j$ to be spatially correlated, then we would potentially have an additional problem for inference.  The intuition is that, in this case, an average of $N_1$ observations of $\mu_j(f)$ for the treated units would be less informative about $\mu_1^e(f)$ than the same average if  $\mu_j(f)$ were independent across $j$. As a consequence, estimated standard errors that ignore this spatial correlation would be under-estimated, which  would lead to over-rejection.  Again, this problem becomes less relevant if the second moment of the distribution of $\nabla \lambda$ is smaller.

\subsubsection{An alternative model in which $F \rightarrow \infty$} \label{Appendix_F}

In Section \ref{sec_LFM}, we consider a linear factor model for the spatial correlation in which the number of factors, $F$, is fixed. While this allows for a rich variety of spatial correlation structures, it would be harder to encompass settings in which, for example, the error is strongly mixing in the cross section. We consider here a stylized example for the spatial correlation, which can also be described as a linear factor model, but in which the number of factors increases when $N_1,N_0 \rightarrow \infty$.  We show that, as in Corollary \ref{Corollary_local},  we also have that (i) ignoring spatial correlation and relying on CRVE generally  leads to over-rejection, and (ii) the over-rejection is stronger when the variance of the difference between the  post- and pre-treatment averages of the common factors is relatively large. 

Consider a simple example in which  we have $N_1/2$ common factors $\lambda_t(f)$, $f=1,...,N_1/2$ and $N_0/2$ common factors $\delta_t(f)$, $f=1,...,N_0/2$. We consider the treatment assignment as fixed, and partition the set of treated units, $\mathcal{I}(1)$, in $N_1/2$ mutually exclusive pairs, $\Lambda_1,...,\Lambda_{N_1/2}$. Likewise, we divide the set of control units, $\mathcal{I}(0)$, in $N_0/2$ mutually exclusive pairs, $\Gamma_1,...,\Gamma_{N_0/2}$.  Potential outcomes are given by 
\begin{eqnarray} \label{LFM_appendix}
 \begin{cases} Y_{jt}(0) = \theta_j + \gamma_t+  \sum_{f=1}^{N_1/2} \lambda_t(f) 1\{ j \in \Lambda_f \} +  \sum_{f=1}^{N_0/2} \delta_t(f) 1\{ j \in \Gamma_f \} + \epsilon_{jt}  \\  Y_{jt}(1) = \alpha +Y_{jt}(0) . \end{cases}
\end{eqnarray}

Therefore, this model for the  potential outcomes follow a linear factor model as the one in equation \ref{LFM}. The main difference is that we allow the number of factors to increase with $N$, and that we impose a structure in which units are divided into pairs that are spatially correlated, but independent across pairs. We assume for simplicity that treatment effects are homogeneous, but all conclusions remain the same if we allow for heterogeneous treatment effects, as we do in Section \ref{problem}. This analysis is conditional on treatment assignment and on the sequence of factor loadings (in this case, the pairs in which each unit belongs), and we impose the following assumptions.

\begin{ass}
\normalfont \label{Assumption_appendix}

(a) $ \{\epsilon_{j1},...,\epsilon_{jT}\}_{\mathcal{I}_0 \cup \mathcal{I}_1}$ is mutually independent across $j$, and identically distributed within treated and control units; (b)  $\{(\lambda_1(f),...,\lambda_T(f))\}_{f=1}^{N_1/2}$ is iid,  $\{(\delta_1(f),...,\delta_T(f))\}_{f=1}^{N_0/2}$ is iid, and these variables are mutually independent; (c) all random variables have finite fourth moments, (d)  $\mathbb{E}[ \nabla \epsilon_j] = 0$ for all $j$,  $\mathbb{E}[\nabla \lambda(f) ] = 0$ for all $f=1,...,N_1/2$, and  $\mathbb{E}[\nabla \delta(f) ] = 0$ for all $f=1,...,N_0/2$.

\end{ass}

Assumption \ref{Assumption_appendix}(a) allows for arbitrary serial correlation in the errors and for arbitrary heteroskedasticity with respect to treatment assignment. Assumption \ref{Assumption_appendix}(d) guarantees that the TWFE estimator is unbiased. Note that we do not need to impose any assumption on $\theta_j$ and $\gamma_t$, because these factors are eliminated by the fixed effects. Therefore, the TWFE estimator eliminates $\theta_j$ and $\gamma_t$ (which may potentially be correlated with treatment assignment), but does not eliminate all of the spatial correlation structure associated with $\{\lambda_t(f)\}_{f=1,...,N_1/2}$ and $\{\delta_t(f)\}_{f=1,...,N_0/2}$. This remaining factor structure does not generate bias given Assumption \ref{Assumption_appendix}(d), but may be problematic for inference if it generates relevant spatial correlation. 

Let $\sigma_\lambda^2 = var(\nabla \lambda(f))$,  $\sigma_\delta^2 = var(\nabla \delta(f))$, and $\sigma_\epsilon^2(w) = var(\nabla \epsilon_j(w))$ for $j \in \mathcal{I}_w$, $w \in \{0,1\}$. Recall that we are considering treatment assignment as fixed in this setting. Therefore, the variance of the TWFE estimator is given by
\begin{eqnarray} \label{Appendix_variance}
var(\hat \alpha) = \frac{2}{N_1} \sigma_\lambda^2 + \frac{2}{N_0} \sigma_\delta^2 + \frac{1}{N_1} \sigma_\epsilon^2(1) + \frac{1}{N_0} \sigma_\epsilon^2(0).
\end{eqnarray}

We consider the asymptotic  behavior of the DID estimator and of  CRVE in this setting when $N_1$ and $N_0 \rightarrow \infty$.

\begin{prop}
 \label{Proposition_appendix}
Consider a setting in which potential outcomes follow equation (\ref{LFM_appendix}). Treatment allocation is fixed, and starts after periods $t^\ast$ for the treated units. Assumption \ref{Assumption_appendix} holds. Then, as  $N_1$ and $N_0  \rightarrow \infty$,
\begin{eqnarray}
\sqrt{N}(\hat \alpha - \alpha) \buildrel d \over \rightarrow N\left(0, \frac{2}{c}\sigma^2_\lambda + \frac{2}{1-c}\sigma^2_\delta + \frac{1}{c} \sigma^2_\epsilon(1) + \frac{1}{1-c} \sigma^2_\epsilon(0)\right).
\end{eqnarray}  
where $N_1/N=c$. Moreover, if $\alpha = 0$, 
\begin{eqnarray}
t =\frac{\hat \alpha}{\sqrt{\widehat{var(\hat \alpha)}_{\tiny \mbox{Cluster}}} } \buildrel d \over \rightarrow N\left(0, 1 + \frac{\frac{1}{c}\sigma^2_\lambda + \frac{1}{1-c}\sigma^2_\delta }{\frac{1}{c}\sigma^2_\lambda + \frac{1}{1-c}\sigma^2_\delta + \frac{1}{c} \sigma^2_\epsilon(1) + \frac{1}{1-c} \sigma^2_\epsilon(0)} \right).
\end{eqnarray}

\end{prop}

\begin{proof}

Note that 
\begin{eqnarray}
\hat \alpha = \alpha + \frac{2}{N_1} \sum_{f=1}^{N_1/2} \nabla \lambda(f) -  \frac{2}{N_0} \sum_{f=1}^{N_0/2} \nabla \delta(f) + \frac{1}{N_1} \sum_{j \in \mathcal{I}_1} \nabla \epsilon_j -  \frac{1}{N_0} \sum_{j \in \mathcal{I}_0} \nabla \epsilon_j.
\end{eqnarray}  

Therefore, applying the central limit theorem, we have
\begin{eqnarray} \label{as_dist}
\sqrt{N}(\hat \alpha - \alpha) \buildrel d \over \rightarrow N\left(0, \frac{2}{c}\sigma^2_\lambda + \frac{2}{1-c}\sigma^2_\delta + \frac{1}{c} \sigma^2_\epsilon(1) + \frac{1}{1-c} \sigma^2_\epsilon(0)\right).
\end{eqnarray}

Now the OLS residuals from TWFE DID regression are such that, for $j \in \mathcal{I}_1$, and $j \in \Lambda_f$,  
\begin{eqnarray}
\widehat{W}_j = \nabla Y_{j} - \frac{1}{N_1} \sum_{k \in \mathcal{I}_1} \nabla Y_{j} = \nabla \lambda(f) + \nabla \epsilon_{j} -  \frac{2}{N_1}\sum_{f' =1}^{N_1/2}   \nabla \lambda(f') + \frac{1}{N_1} \sum_{k \in \mathcal{I}_1} \nabla \epsilon_{k}.
\end{eqnarray}  

Given Assumption \ref{Assumption_appendix}, 
\begin{eqnarray} 
\frac{1}{N_1} \sum_{j \in \mathcal{I}_1} \widehat{W}_j^2  =   \sigma_\lambda^2 + \sigma_\epsilon^2(1) + o_p(1).
\end{eqnarray}  

Using similar calculations for the control units, we have that, up to a degrees-of-freedom correction, 
\begin{eqnarray}  \label{Proof_variance}
N \widehat{var(\hat \alpha)}_{\tiny \mbox{Cluster}} &=& N \left[\frac{1}{N_1} \left( \frac{1}{N_1} \sum_{j \in \mathcal{I}_1} \widehat{W}_j^2 \right) + \frac{1}{N_0} \left( \frac{1}{N_0} \sum_{j \in \mathcal{I}_0} \widehat{W}_j^2 \right)  \right] \\
&=& \frac{1}{c} \sigma_\lambda^2 + \frac{1}{1-c} \sigma_\delta^2 + \frac{1}{c} \sigma_\epsilon^2(1) + \frac{1}{1-c} \sigma_\epsilon^2(0) + o_p(1).
\end{eqnarray}  

Combining equations \ref{as_dist} and \ref{Proof_variance} finishes the proof. 
\end{proof}

Proposition \ref{Proposition_appendix} shows that, in this setting, the TWFE estimator is asymptotically normal. However,  CRVE will underestimate the asymptotic variance of the TWFE estimator. Moreover, if we assume  $\lambda_t(f)$ and $\delta_t(f)$ are serially positively correlated, with stronger dependence relative to the idiosyncratic shocks, then the distortion in the variance due to spatial correlation would be less relevant if we consider a shorter  distance between the initial and final periods. This is essentially the same conclusion from Corollary \ref{Corollary_local}, but for a spatial correlation model based on a linear factor model in which the number of factors increases with $N$. This allows for settings in which the spatial correlation is strongly mixing, as considered by \cite{ferman2020inference}.

\subsubsection{Model-based versus Design-based uncertainty } \label{appendix_design}

We provide a simple example  showing that the main intuitions in this paper would also apply if we consider a design-based approach, in which uncertainty comes only from the treatment allocation \citep{Finite_pop, NBERw24003, ATHEY2021,rambachan2020designbased}. In this case, we condition on a realization of the potential outcomes, and spatial correlation is captured by considering that the treatment allocation is spatially correlated. \cite{NBERw24003} also consider the case in which treatment allocation is spatially correlated. The difference is that here we are fundamentally interested in the case in which it is not possible to cluster at the treatment assignment level. This may be the case, for example, because the researcher does not have information on the relevant distance metric, or because there are too few clusters to rely on CRVE at the assignment level. 

Consider a very simple example where states $j=1,...,N$ are partitioned into equally-sized groups of states $\Lambda_1,...,\Lambda_F$, and potential outcomes are given by 
\begin{eqnarray} 
 \begin{cases} Y_{jt}(0) = \theta_j + \gamma_t+  \sum_{f=1}^{F} \lambda_t(f) 1\{ j \in \Lambda_f \} + \epsilon_{jt}  \\  Y_{jt}(1) = \alpha_{jt} +Y_{jt}(0) . \end{cases}
\end{eqnarray}

Note that this is a particular case of the  potential outcomes model determined by equation (\ref{LFM}). We think of that as a ``super-population'' model where the finite population is drawn. Therefore, when we consider such design-based approach, we condition  on the realizations of $\theta_j$, $\gamma_t$, $\lambda_t(f)$ for $f=1,...,F$, and $\epsilon_{jt}$, for all states and for all periods.  For simplicity, we assume treatment effects are homogeneous, and consider the case in which $\alpha_{jt} = 0$ for all $j$ and $t$. In this case, our estimand, which is the finite-population analogue to the average treatment effect, is equal to zero.  

To capture spatial correlation problems, we consider that treatment allocation is such that $F/2$ groups of states are randomly allocated into treatment, and then all states in these groups receive treatment. Therefore, the DID estimator is unbiased over the treatment assignment distribution.\footnote{Let $\pi_j$ be the marginal probability of treatment for state $j$. From the results derived by \cite{rambachan2020designbased}, it is clear that the DID estimator is unbiased over the treatment assignment distribution, because $\pi_j=1/2$ for all $j$. { More generally, \cite{rambachan2020designbased} show that the DID estimator is unbiased over the randomization distribution if $\sum_{j=1}^N (\pi_j - \bar \pi) (\nabla Y_{j}(0))=0$. Considering an alternative randomization distribution that satisfies this condition on the marginal probabilities of treatment assignment does not change our main conclusions. } } 
 The problem we want to evaluate is whether  researchers would face relevant inference distortions if they cluster they standard errors at the state level (instead of clustering at the $F$ groups of states). In other words, in this design-based approach, we consider the case in which treatment was assigned at the ``groups of state'' level, but the researchers proceeded with an inference method that would be asymptotically valid if treatment were assigned at the state level. As mentioned above, this can be the case because they were unaware that treatment was assigned at a ``groups of states'' level.

From Lemma 5 from \cite{IK}, the exact  variance of the DID estimator under this spatially correlated treatment assignment, conditional on the potential outcomes, is given by
\begin{eqnarray} 
\mathbb{V}_{corr} = \frac{4}{F(F-2)} \sum_{f=1}^F \left( \nabla \lambda(f) - \nabla \bar \lambda + \nabla \bar \epsilon_f - \nabla \bar \epsilon  \right)^2,
\end{eqnarray}
where $ \nabla \bar \lambda = \frac{1}{F} \sum_{f=1}^F  \nabla \lambda(f) $, $ \nabla \bar \epsilon_f = \frac{1}{N/F} \sum_{j \in \Lambda_f} \nabla \epsilon_j$, and $ \nabla \bar \epsilon = \frac{1}{N} \sum_{i=1}^N \nabla \epsilon_i$.

In contrast, if we considered that treatment was assigned with no spatial correlation, then the variance would be given by
\begin{eqnarray} 
\mathbb{V}_{uncorr} = \frac{4}{N(N-2)} \sum_{j=1}^N \left(  \sum_{f=1}^F  [\nabla \lambda(f) 1\{ j \in \Lambda_f \} ] - \nabla \bar  \lambda + \nabla  \epsilon_j - \nabla \bar \epsilon  \right)^2.
\end{eqnarray}

Note that CRVE at the state level would approximate $\mathbb{V}_{uncorr}$. Therefore, we consider the extent to which $\mathbb{V}_{uncorr} $ underestimates $\mathbb{V}_{corr} $.
\begin{eqnarray}  \nonumber \label{formula_design}
\mathbb{V}_{corr}  - \mathbb{V}_{uncorr} &=& \frac{1}{F} \sum_{f=1}^F \left( \nabla \lambda(f) - \nabla \bar \lambda \right)^2 \left[ \frac{4}{F-2} - \frac{4}{N-2} \right] \\ 
&&+  \frac{4}{N(N-2)} \sum_{j=1}^N \left(  \nabla  \epsilon_j - \nabla \bar \epsilon  \right)^2 \left[ \frac{F}{F-2}\frac{N}{N-2} -1 \right] \\ \nonumber
&& +  \frac{4}{N} \sum_{j=1}^N \left(  \sum_{f=1}^F [\nabla \lambda(f) 1\{ j \in \Lambda_f \}]  - \nabla \bar \lambda  \right) \left(  \nabla  \epsilon_j - \nabla \bar \epsilon \right) \left[ \frac{1}{F-2} - \frac{1}{N-2} \right] \\ \nonumber
&& + \frac{F}{F-2}\frac{N-2}{N} \frac{4}{N(N-2)}  \sum_{f=1}^F \sum_{i \neq j, i,j \in \Lambda_f} (\nabla \epsilon_i -  \nabla \bar\epsilon) (\nabla \epsilon_j-  \nabla\bar \epsilon).
\end{eqnarray}

We consider first a case in which  $F$ is fixed and $N \rightarrow \infty$. All conclusions remain valid if we consider a setting in which both $F$ and $N \rightarrow \infty$, similarly to what we show in Appendix \ref{Appendix_F} for the model-based case. We discuss this case below. 

The literature on design-based uncertainty imposes assumptions on the sequence of potential outcomes of the finite populations. We can think that there is a super-population in which we draw such finite population. In this case, we think about potential outcomes in this super-population as random variables. In this super-population, we assume $\epsilon_{jt}$ are independent across $j$, as we do in Section \ref{problem}, implying that, when $N \rightarrow \infty$, the last three  terms in equation (\ref{formula_design}) converge almost surely to zero. Therefore, to be consistent with the assumptions on the super-population, we assume that the sequence of finite populations is such that these three terms converge to zero (in this case, these terms are non-stochastic sequences).

The first term of equation (\ref{formula_design}), however, converges to $ \frac{4}{F(F-2)} \sum_{f=1}^F \left( \nabla \lambda(f) - \nabla \bar \lambda \right)^2 >0$ when $N \rightarrow \infty$. Importantly,  if the variance of $\nabla \lambda(f)$ is lower in the super-population, then the probability that we condition on a realization of $(\nabla \lambda(1),...,\nabla \lambda(F))$ such that $\frac{4}{F(F-2)} \sum_{f=1}^F \left( \nabla \lambda(f) - \nabla \bar \lambda \right)^2 $ is larger would be lower. Therefore, we reach exactly the same conclusion  from Proposition \ref{Proposition}, where we considered a model-based uncertainty. In this setting, the estimator would not generally be consistent and asymptotically normal, similarly to what we find in Proposition \ref{Proposition}. However, the extent to which we underestimate the variance, and to which we depart from asymptotic normality, depends on the term  $ \frac{4}{F(F-2)} \sum_{f=1}^F \left( \nabla \lambda(f) - \nabla \bar \lambda \right)^2 >0$, which we expect to be smaller when the variance of $\nabla \lambda(f)$ is smaller in the super-population.

  The case with $F \rightarrow \infty$ is similar, with the difference that in this case we would find $F( \mathbb{V}_{corr}  - \mathbb{V}_{uncorr}) \rightarrow K \times \mbox{lim}_{F\rightarrow\infty} \frac{1}{F} \sum_{f=1}^F \left( \nabla \lambda(f) - \nabla \bar \lambda \right)^2 $, for some constant $K$. This constant is greater than zero if  $F/N \rightarrow c \in [0,1)$.  In this case, the estimator would be consistent and asymptotically normal, but we would have over-rejection, because we would  under-estimate the standard errors.  Again, this scenario is consistent with the main conclusions of the paper about when we should expect spatial correlation to be more relevant. This scenario parallels the results presented in  Appendix \ref{Appendix_F}. The case $F/N \rightarrow 1$ implies $K=0$, so the variance is not underestimated. This happens when we have a very large number of groups of states with only one states, which essentially means that we do not have much spatial correlation, so it is reasonable that the variance is not underestimated in this case.

If we relax the assumption that treatment effects are homogeneous, then, following \cite{Finite_pop}, \cite{NBERw24003}, and \cite{rambachan2020designbased}, we should expect CRVE to be conservative relative to $\mathbb{V}_{uncorr}$. While this could partially offset part of the underestimation of $\mathbb{V}_{corr} $ when spatial correlation is not taken into account, the same conclusions about when spatial correlation should lead to more significant problems for inference would still apply. 

Finally, we note that in an earlier version of this paper we consider simulations based on such design-based approach for inference \citep{Ferman_OLD}.

\subsection{Alternative estimators  } \label{Appendix_LFM}

We show that  alternative estimators designed for settings in which the linear factor model may affect the counterfactual trends would generally not be feasible (or would present some limitations) in the type of applications we consider. 

Since we consider  a setting with a fixed number of periods, it would not be possible to use an interactive fixed effects estimator, as proposed by \cite{Bai} and \cite{Gobillon2016}, unless we impose strong assumptions on the errors (e.g., \cite{Bai2003} and \cite{anderson1984introduction}).  This linear factor model  structure is also considered in the synthetic control (SC) literature \citep{Abadie2010,FP_QE,Ferman_JASA}. However, the conditions in which the SC estimator takes into account such linear factor structure will generally rely on a large number of periods, while we focus on  the case in which $T$ is fixed. 

Finally, while the common correlated effects (CCE) estimator proposed by \cite{Pesaran} may work in a fixed-$T$ asymptotics \citep{Westerlund}, note that a standard DID setting in which all treated units start treatment in the same period would not satisfy the standard assumptions for this estimator. The reason is that the rank condition in Assumption 5 from \cite{Pesaran} would not hold. The averages of $d_{jt}$ would be 0 in the pre-treatment periods and $\rho$ (the proportion of treated) in the post-treatment periods. Therefore, when we regress $(d_{j1},...,d_{jT})$ on  $(\bar d_{1},...,\bar d_{T})$, the residuals will be a vector of zeros  for all $j$. 

If there is variation in treatment timing, then Assumption 5 from \cite{Pesaran} would hold. While previous papers on the CCE estimator imposed assumptions that would preclude the case of a treatment dummy, \cite{Brown_CCE} consider a setting in which this would be allowed. However, there would still be some important limitations in this case. First, the CCE estimator imposes restrictions on the dimension of the linear factor model. In order to illustrate that, we consider a simple simulation study, which is presented in Appendix Table \ref{Table_CCE}. While the CCE estimator works in taking the spatial correlation into account in the DGP's in which the dimension of the linear factor model  is only 1 or 2, it leads to large over-rejections (at the same level as the TWFE estimator) when this dimension is 3. Also, considering settings in which $\mu_1^e = \mu_0^e$, so that inference based on CRVE at unit level is valid for both the CCE and TWFE estimators, we find that there is a loss in precision in using CCE relative to TWFE. 

In addition to that, while the CCE estimator allows for heterogeneous treatment effects in the cross-section (by considering the case of random slopes), it may have problems in aggregating heterogeneous effects if there is also heterogeneity across time. More specifically,  the CCE estimand may be negative even when treatment effects are always positive, a problem that has been extensively discussed in the DID literature for the TWFE estimator \citep{Bacon,chaisemartin2018twoway}, but not for the CCE estimator.\footnote{As an example, consider a setting with $T=4$ in which $400$ units are divided into four groups. Group 1 starts treatment after period one, with treatment effects of $(1,10,100)$ on periods 2 to 4, group 2 starts treatment after period two, with treatment effects of $(1,100)$ on periods 3 and 4, while group 3 starts treatment only in the last period, with treatment effect of 100. Group 4 is never treated. The CCE estimand in this case is -3.7, despite the fact that treatment effects are always positive.}  Moreover, since the CCE estimator does not work for settings in which there is no variation in treatment timing,  standard solutions to this problem that have been considered in the DID literature cannot be directly applied for the CCE estimator, since it would not work when we consider simple $2\times2$ DID settings (see Remark \ref{Remark_estimators}). 

Overall, while CCE estimator provides an interesting alternative to take spatial correlation into account in some settings in which inference for the TWFE estimator (or  alternative DID estimators) would be problematic, there are some important limitations and trade-offs involved in using this approach. One alternative when there is variation in treatment timing is to  consider both the CCE estimator and other alternatives recommended in Section \ref{recommendations} to mitigate inference problems due to spatial correlation, and check whether the results are robust to all of those alternatives.

\subsection{Time frame \& size distortions }
\label{appendix_dependence}

\subsubsection{TWFE with $T$ periods} 
\label{Appendix_time_frame_proof}

Consider the case in which we have $T$ periods, and $t^\ast = T/2$. Assume  that $\xi_t (\mu_1^e - \mu_0^e)$ follow an AR(1) process with serial correlation $\rho_\xi$, while $\epsilon_{jt}$, conditional on either $D_j = 0$ or $D_j = 1$, follows an AR(1) process with serial correlation $\rho_\epsilon$, and that $0 \leq  \rho_\epsilon < \rho_\xi < 1$. This assumption means that the spatially correlated shocks are also more serially correlated than the idiosyncratic shocks. We show that, under these conditions, size distortions when spatial correlation is ignored are increasing in $T$. 
  
  Consider a random variable $X_t$ that  follows an  AR(1)  process, $X_t = \rho X_{t-1} + \nu_t$, where $\nu_t$ is iid with variance $\sigma_\nu^2$. Since  $X_t$ is stationary, then $var(X_t) = var(X_{t-1}) = \frac{1}{1-\rho^2} \sigma^2_\nu$. Assume we have $T$ periods and define $\bar X_{post}$ as the average in the first half of the periods, and $\bar X_{pre}$ as the average in the second half of the periods, with $\nabla X = \bar X_{post} - \bar X_{pre} $. In this case,
\begin{eqnarray}
\mathbb{E}[( \nabla X)^2] = \frac{4}{T^2(1-\rho)^2}\left[T-2\frac{\rho}{1-\rho^2}(3-\rho^{T/2})(1-\rho^{T/2})\right] \sigma^2_\nu.
\end{eqnarray}

Consider now the asymptotic distribution of the t-statistic presented in Equation \ref{t_stat}. Note that this implies that the over-rejection due to spatially correlated shocks is increasing in the ratios $\frac{(\mu_1^e - \mu_0^e)' \Omega (\mu_1^e - \mu_0^e)}{\sigma_\epsilon^2(w)} $ for $w \in \{0,1\}$. The numerator is the variance of $\nabla \xi (\mu_1^e - \mu_0^e)$, while the denominator is the variance of $\nabla \epsilon_{j}$.\footnote{Assume for simplicity a model with homogeneous treatment effects and homoskedasticity.}  Now define
\begin{eqnarray} \label{phi}
\phi(\rho_\xi,\rho_\epsilon,T) = K \left[ \frac{T-2\frac{\rho_\xi}{1-\rho_\xi^2}(3-\rho_\xi^{T/2})(1-\rho_\xi^{T/2})}{T-2\frac{\rho_\epsilon}{1-\rho_\epsilon^2}(3-\rho_\epsilon^{T/2})(1-\rho_\epsilon^{T/2})} \right].
\end{eqnarray}

For some constant $K>0$, this formula presents the ratio $\frac{(\mu_1^e - \mu_0^e)' \Omega (\mu_1^e - \mu_0^e)}{\sigma_\epsilon^2(w)} $ when $\nabla \lambda (\mu_1^e - \mu_0^e)$ is AR(1) with serial correlation $\rho_\xi$, while  $\nabla \epsilon_{j}$ is AR(1) with serial correlation $\rho_\epsilon$. Considering all combinations of $(\rho_\xi,\rho_\epsilon)$ such that $0 \leq \rho_\epsilon < \rho_\xi < 1$ in $0.01$ intervals, and $T \in \{ 2,4,6, \hdots, 200\}$, we find numerically that $\phi(\rho_\xi,\rho_\epsilon,T) $ is increasing in $T$ for all of these combinations of parameters. Therefore, if the common factors are positively serially correlated with a stronger serial correlation relative to the idiosyncratic shocks, then we should expect that spatial correlation leads to less distortions for inference when we  consider shorter time frames.

\subsubsection{Dynamic DID}

We now continue to consider the same assumptions as in Appendix \ref{Appendix_time_frame_proof} on $\xi_t (\mu_1^e - \mu_0^e)$ and $\epsilon_{jt}$. However, in this case we consider the two-period DID estimator with pre-treatment period $t_0 \leq t^\ast$, and post-treatment period $t_0 + \tau > t^\ast$. Again, we show that, if $0 \leq  \rho_\epsilon < \rho_\xi < 1$, then size distortions will be increasing in $\tau$.

First, let $X_t$ be a stationary  AR(1) random variable, $X_t = \rho X_{t-1} + \nu_t$, where $\nu_t$ is iid with mean zero and variance $\sigma^2$. It follows that $var(X_t) = \frac{\sigma^2}{1-\rho^2}$. Now let $\Delta_\tau X_t = X_{t} - X_{t-\tau}$. In this case, we have that $var(\Delta_\tau X_t) = \frac{2\sigma^2}{1-\rho^2} \left( 1 - \rho^\tau \right)$.

Now consider the case in which $0 \leq \rho_\epsilon < \rho_\xi < 1$. From Corollary \ref{Corollary_local}, and given the formula for $var(\Delta_\tau X_t)$,  over-rejection will be increasing in $ (1 - \rho_\xi^\tau)/(1 - \rho_\epsilon^\tau)$. Therefore, in order to show that over-rejection increases in $\tau$, it suffices to show that $$ \Phi(\tau; \rho_\epsilon,\rho_\xi) \equiv \frac{1 - \rho_\epsilon^{\tau+1}}{1 - \rho_\xi^{\tau+1}} -  \frac{1 - \rho_\epsilon^{\tau}}{1 - \rho_\xi^{\tau}} < 0$$
for any $\tau \in \mathbb{N}$ and $0 \leq \rho_\epsilon < \rho_\xi < 1$. 

In this case, we can show that analytically. First, note that $\Phi(\tau; 0,\rho_\xi) < 0$ and $\Phi(\tau; \rho_\xi,\rho_\xi) = 0$. Therefore,  we only need to show that $\frac{\partial}{\partial \rho_\epsilon}\Phi(\tau; \rho_\epsilon,\rho_\xi) > 0$ for $\rho_\epsilon \in [0,\rho_\xi]$ to conclude the proof. For some $C>0$, we have 
\begin{eqnarray*}
\frac{\partial}{\partial \rho_\epsilon}\Phi(\tau; \rho_\epsilon,\rho_\xi) &=& C  \left[ \tau (1 - \rho_\xi^{\tau+1}) - (\tau + 1) \rho_\epsilon (1 - \rho_\xi^\tau)  \right]   >  C  \left[ \tau (1 - \rho_\xi^{\tau+1}) - (\tau + 1) \rho_\xi (1 - \rho_\xi^\tau)  \right] \\
& = & C  \left[ \tau + \rho_\xi^{\tau + 1} - \tau \rho_\xi - \rho_{\xi}   \right].
\end{eqnarray*}

Now note that the expression $\kappa(\rho_\xi) = \tau + \rho_\xi^{\tau + 1} - \tau \rho_\xi - \rho_{\xi}  $ is such that $\kappa(0) = \tau >0$ and $\kappa(1) = 0$. Moreover, $\kappa'(\rho_\xi) = (\tau + 1)(\rho_\xi^\tau - 1) \leq 0$ for $\rho_\xi \in [0,1]$, implying that $\kappa(\rho_\xi) \geq 0$ for any $\rho_\xi \in [0,1]$, and, therefore, implying that $\frac{\partial}{\partial \rho_\epsilon}\Phi(\tau; \rho_\epsilon,\rho_\xi) > 0$ when $0 \leq \rho_\epsilon < \rho_\xi < 1$. Combining all those pieces, we have that $ \Phi(\tau; \rho_\epsilon,\rho_\xi) < 0$ for all $\tau \in \mathbb{N}$ and   $0 \leq \rho_\epsilon < \rho_\xi < 1$.

\subsection{More Details on MC Simulations} \label{Appendix_MC}

\subsubsection{Simulations with ACS data} \label{A_ACS}

We describe in more details the procedures used for the simulations in Section \ref{Sim_ACS}. We use information on the ACS from 2005 to 2019, and we restrict the sample for women aged from 25 to 50. With this sample, we aggregate the data in (year $\times$ PUMA) cells. 

For each $T \in \{ 2,\hdots,15 \}$, we restrict the sample to a time frame with $T$ periods. Based on the industry composition of the workers in the initial time period, we use a $k$-means clusters to partition the PUMA's into 50 groups with similar industry compositions. Then we calculate $\nabla Y_{j,s,k}$ for PUMA $j$, in state $s$, and industry cluster $k$.  

We assume a model in which 
\begin{eqnarray}
cov(\nabla Y_{j,s,k},\nabla Y_{j',s',k'}) &=& \sigma^2 1\{j=j' ~ \& ~ s=s' ~ \& ~ k = k'\} + \sigma^2_{\mbox{\tiny state }} 1\{s=s' \} \\
&&  + \sigma^2_{\mbox{\tiny ind }} 1\{k=k' \}  + \sigma^2_{\mbox{\tiny both }} 1\{s=s' ~ \& ~ k=k' \}.
\end{eqnarray}

Therefore, this model assumes that the correlation between PUMA's that are not in the same state or in the same industry cluster is zero. However, PUMA's in the same state and/or in the same industry cluster may be spatially correlated.

We estimate this model for the covariance matrix using the ACS data. When we are able to consider a sample with $T$ periods for different initial years $t_0$, we estimate the correlation structure for each $t_0$, and then aggregate the information from all $t_0$.  Given the estimated correlation structure, we consider a Gaussian model with mean zero and  with such correlation structure for the simulations.  

We consider two different treatment assignments. In the first one, PUMA's are randomly assigned into treatment. In the second one, we first select the clusters of industries that are more concentrated into manufacturing. Then, we assign treatment with 90\% probability for those PUMA's, and 10\% probability for those in clusters of industries with lower concentration of manufacturing. 

{  Note that in both cases we have that the DID estimator is unbiased given the DGP considered in the simulations. In order to check whether considering a DGP that imposes parallel trends provides a good approximation for the original data, in Appendix Figure \ref{App_figure_PT} we present the time series of average log wages (based on the original data)  for the PUMAs that are more concentrated in manufacturing, and for those that are not. This allows for within-state correlations, and also for correlations between PUMA's in the same industry cluster (even if they are in different states).\footnote{If anything, those standard errors underestimate the true ones in case there is correlation between PUMAs in different industry clusters. }  We also  test  the null that these two groups of PUMAs have the same linear time trend, and we fail to reject the null (p-value = 0.468). Therefore,  we do not find evidence that these two sets of PUMAs present systematically different trends. This provides support for the construction of these simulations assuming parallel trends. }

\subsubsection{Simulations with CPS data} \label{A_CPS}

We describe in more details the procedures used for the simulations in Section \ref{Sim_CPS}. We use information on the CPS from 1990 to 2018, and we restrict the sample for women aged from 25 to 50. With this sample, we aggregate the data in (year $\times$ state $\times$ age) cells.\footnote{All results from the simulations are very similar if we consider data from 1979 to 2018. We consider a more restricted sample because the parallel trends assumption is more plausible when we consider this range.}

For each $(T,\delta_{\mbox{\tiny age}})$, we construct a DGP for $\mathbf{W}$, which is a $(51  \delta_{\mbox{\tiny age}} \times 1)$ vector  with elements $W_{j,a}$, which is the post-pre outcomes for state $j \in \{ 1,...,51\}$ and age group $a \in \{1,...,\delta_{\mbox{\tiny age}}\}$ (running from the youngest to the oldest age group in the sample). Note that, based on Equation \ref{alpha}, we focus on the post-pre treatment averages of the outcome, so that we only need to model the cross-section correlations of $W_{j,a}$. For a given $(T,\delta_{\mbox{\tiny age}})$, we construct this DGP  in the following way:

\begin{enumerate}

\item We consider all combinations of initial year $t_0$ and youngest age $A_0$ such that we have information on the $t_0 + T-1$ period and on the  $A_0 + \delta_{\mbox{\tiny age}}-1$ age group. 

\item For each of these combinations, we calculate $\nabla Y_{j,A} (t_0)$ as the post-pre average differences in log wages for state $j$ and age group $A \in \{25,...,50\}$ when we restrict the sample for  years $t_0$ to $t_0 + T -1$, and treatment starts in the middle of this time range. 

\item We calculate the $(j,A)$-specific mean of $\nabla Y_{j,A} (t_0)$ across $t_0$, and define $W_{j,A}(t_0) = \nabla Y_{j,A} (t_0) - \overline{\nabla Y_{j,A}}$.

\item For each $(t_0, A_0)$, let $\mathcal{L}(t_0, A_0)$ be a $(51  \delta_{\mbox{\tiny age}} \times 1)$ vector that collects information on \\ $ \{W_{j,A_0}(t_0), \hdots , W_{j,A_0 +\delta_{\mbox{\tiny age}}-1}(t_0)  \}_{j=1}^{51}$. That is, $\mathcal{L}(t_0, A_0)$ contains  information on $W$ for $\delta_{\mbox{\tiny age}}$ age groups for each of the 51 states.

\item For our DGP, we will consider a random sample  of  $\mathcal{L}(t_0, A_0)$ across all $(t_0, A_0)$ that are feasible given $(T,\delta_{\mbox{\tiny age}})$ to generate a vector of outcomes $\mathbf{W}$ with elements $W_{j,a}$, which is the outcome for state $j \in \{ 1,...,51\}$ and age group $a \in \{1,...,\delta_{\mbox{\tiny age}}\}$.

\end{enumerate}

Note that, given this procedure, $\mathbb{E}[W_{j,a}] = 0$ for all $j$ and $a$ for this DGP. Therefore, the null  that average treatment effect is zero is valid. Importantly, this DGP brings from the original data a spatial correlation structure in which age groups that may be  closer are spatially correlated, whether or not they belong to the same state. Moreover, we may even have a spatial correlation in this DGP that is not weakly mixing, which is important in this set of simulations, since we want to consider the case in which there are only a few age groups. Finally, note that we generate different spatial correlation structures depending on $T$, and this is also based on the serial/spatial correlation structure in original data. 

Then, for each realization of $\mathbf{W}$, the DID estimator is simply the difference between the averages of the age groups above the median and the average of those below the median (since $W$ is already the post-pre time difference, we only need to take one additional difference to compute the DID estimator). We focus on CRVE at the state level, which takes into account that age groups within the same state may be correlated, but fails to take into account that similar age groups across states may be correlated as well. Note that the point estimate and the standard errors in this regression that we run would be numerically equivalent to the DID estimator with state $\times$ age $\times$ time data when we consider clustered standard errors at the state level.

Note that the construction of the DGP considered in these simulations impose that  the parallel trends assumption is valid, meaning that older and younger age groups do not follow systematically different time patterns. In order to check whether considering a DGP that imposes parallel trends provides a good approximation for the original data, in Appendix Figure \ref{App_fig_CPS} we present the time series of average log wages for younger age groups (ages 25 to 37) versus older age groups (ages 38 to 50). We consider standard errors that allow for correlation between individuals within the same state, and also between individuals in close  age groups (even if  in different states).  We do not observe systematically different trends for these two groups.  We  test the null that these two groups follow the same linear trend, and we fail to reject the null (p-value = 0.290). This provides support for the construction of these simulations assuming parallel trends.

\subsubsection{Alternative Inference Methods for the Simulations with CPS data} \label{Appendix_alternatives_CPS}

Since in these simulations the distance metric that generates spatial correlation is known (in this case, age), we consider alternative inference methods that adjust the standard errors for spatial correlation. While it is always recommended to take spatial correlation into account when it is feasible to do so, we show that existing methods that take spatial correlation into account do not work well in the simulations we consider in this section.

We consider first the use of the standard errors proposed by \cite{CONLEY19991}. In order to take into account that we may have correlation when $W_{j,a}$ and $W_{j',a'}$ belongs to the same state ($j = j'$), and also when they do not belong to the state (but have similar ages), we construct a distance matrix $d((j,a),(j',a'))$ which is equal to zero when $j=j'$, and equal to $|a - a'|$ when $j \neq j'$. In this case, for example, if we consider a cut-off for the construction of  \cite{CONLEY19991} between zero and one, then the standard errors would (potentially) take into account within state correlations and between state correlations if we consider the same age group (but not from observations in different states and in different age groups). If the cut-off is between one and two, then we would also take into account correlation across states if the difference in the age group is one, and so on. 

The problem with this approach is that, given the structure of our simulations, we do not have much variation in the age groups to take all of those correlations into account. Consider the  case with $\delta_{\mbox{\tiny age}}=10$, which is the setting in which we have more variation in age groups. If the cut-off is between zero and one, we find rejection rates larger than the ones we find in Table \ref{Table_CPS}. This happens because not only the standard errors fail to take into account across state correlations when $a \neq a'$, but also because the standard errors are underestimated. Even if we consider an iid normal outcome (so that there is no spatial correlation), we would find 12\% rejection rates for a 5\% nominal level test when we consider a cut-off between zero and one (see \cite{Assessment} for the idea of considering such inference assessment).

When we increase the cut-off (so that we allow for spatial correlation across more observations), the standard errors  become even more underestimated. Again considering the case with iid errors, we would have rejection rates of, for example, 24\% when the cut-off is between one and two, and of 32\% when it is between two and three. Rejection rates become even higher when we consider that those standard errors do not take into account correlation between individuals with more than two or three years distance in terms of age. Moreover, in this last case, the estimated variance is negative in 28\% of the cases, so it is not even possible to calculate the standard errors in those cases.   

Overall, this approach does not work well in this setting, and the reason is that we do not have enough variation in the distance metric. This problem is even more severe when we consider settings with $\delta_{\mbox{\tiny age}}<10$.   In other settings with more variation in the distance metric, however, this can be an interesting alternative.

A recent alternative that has, under some conditions, valid finite sample properties was proposed by \cite{Muller1,Muller2}.  However, it is not trivial to implement their approach in our setting, because we have this structure in which we want to take into account both within state correlations and similar age groups correlations. \cite{Muller2} include an option to consider cluster in the method, but this approach assumes that all observations within a cluster are in the same location. However, since in all states we have all of the age groups, we end up with no variation to estimate the across-state correlations. If we attempt to use their code with cluster at the state level in our setting, the code does not even run. Another alternative we considered was to aggregate the observations at the age-group level, and consider  \cite{Muller2}  for this aggregated data (which takes into account that the age-group aggregates may be correlated). The main difficulty in this case is that this approach relies on pre-specifying a parameter that reflects the maximum average pairwise correlation $(\bar \rho)$. For the case with $(T,\delta_{\mbox{\tiny age}})=(10,10)$, the standard $\bar \rho$ is too small in this application, so the method leads to large over-rejection. If we increase  $\bar \rho$, then rejections go down, but at the cost of having a low powered test in settings in which ignoring spatial correlation would not lead to large distortions, such as when $(T,\delta_{\mbox{\tiny age}})=(1,10)$. Moreover, it is not possible to use this idea when $\delta_{\mbox{\tiny age}} \leq 3$, because there is not enough variation in the data to compute the critical values. 

Finally, we note that even in settings in which it is possible to implement the method proposed by   \cite{Muller1,Muller2}, we may have relevant size distortions if errors are heteroskedastic. For example, consider a simple case in which we have 20 locations, ordered from 1 to 20. In each location, we have 20 observations, and observations in the last $T_1$ locations are treated. We consider the case in which errors are iid normal with mean zero, but with a standard deviation $k$ times larger for the treated unit, where $k \in \{1,1.5,2 \}$. Appendix Figure \ref{Fig_MW} shows that we can have relevant over- or under-rejection in settings in which there is an unequal number of treated and control locations.

\subsection{Pre-testing for spatial correlation problems} \label{pretest}

In settings with more than one pre-treatment period, it is also possible to conduct placebo exercises to test whether spatial correlation is a problem. For example, consider a setting with two pre-treatment periods ($t \in \{ -1,0 \}$) and one post-treatment period ($t \in \{ 1 \}$). In this case, we can consider an estimator for the treatment effect using periods $t \in \{0,1\}$,  $\hat \alpha_{1} = \frac{1}{N_1} \sum_{i \in \mathcal{I}_1} \Delta Y_{i1} - \frac{1}{N_0} \sum_{i \in \mathcal{I}_0} \Delta Y_{i1}$, where for a generic variable $A_t$, $\Delta A_t = A_t - A_{t-1}$, and the pre-treatment periods to test whether inference based on CRVE is reliable. In this case, we would test whether $\hat \alpha_{0} = \frac{1}{N_1} \sum_{i \in \mathcal{I}_1}  \Delta Y_{i0} - \frac{1}{N_0} \sum_{i \in \mathcal{I}_0} \Delta Y_{i0}$ is different from zero. This has been widely considered in the literature as a test for pre-trends (e.g., \cite{Freyaldenhoven}, \cite{Lang}, and \cite{Roth}). In contrast, here we assume that trends are parallel, so $\mathbb{E}[\hat \alpha_0]=0$ and $\mathbb{E}[\hat \alpha_1] = \alpha$, and show that such test can also be informative about whether spatially correlated shocks poses relevant problems for inference.\footnote{In a revised version of his paper  developed concurrently with our paper, \cite{Roth} considers in Appendix D simulations in a setting with stochastic violations of parallel trends that is similar to our setting with spatially correlated shocks.}

 We consider  in detail the case in which potential outcomes are given by equation (\ref{LFM}), but all our results are valid for more general settings. Under Assumptions \ref{rv} to \ref{local0}, and considering that $N\rightarrow\infty$,  we have from Corollary \ref{Corollary_local} that 
 \begin{eqnarray}
\widehat{var(\hat \alpha_\tau)}_{\tiny \mbox{Cluster}} = var\left(\hat \alpha_{\tau}\right) -  (\mu^e_1 - \mu^e_0)' \Omega_\tau(\mu^e_1 - \mu^e_0) + o_p(1),
\end{eqnarray}  
where $\Omega_\tau = \mathbb{E}[\Delta \xi_\tau' \Delta \xi_\tau ]$.

The intuition behind the pre-test for spatial correlation is that, if $\mathbb{E} \left[\Delta \xi_0' \Delta \xi_0 \right] \approx \mathbb{E} \left[\Delta \xi_1' \Delta \xi_1 \right]$, then rejecting the null that $\mathbb{E}[\hat \alpha_0]=0$ would provide evidence that $var\left(\hat \alpha_{0}\right)$ is underestimated when we consider $\widehat{var(\hat \alpha_0)}_{\tiny \mbox{Cluster}}$, which in turn would indicate that $var\left(\hat \alpha_{1}\right)$ is underestimated when we consider $\widehat{var(\hat \alpha_1)}_{\tiny \mbox{Cluster}}$.  If we assume that common factors are stationary, then  $\mathbb{E} \left[\Delta \xi_0' \Delta \xi_0 \right] = \mathbb{E} \left[\Delta \xi_1' \Delta \xi_1 \right]$ and the pre-test would be informative.   

Building on the setup considered by \cite{Roth}, we consider a setting where $(\hat \alpha_{1}, \hat \alpha_{0})$ is jointly normally distributed,
\begin{eqnarray}
\begin{pmatrix}  \hat \alpha_{1} \\ \hat \alpha_{0}  \end{pmatrix} \sim N\left(  \begin{bmatrix}   \alpha \\ 0 \end{bmatrix}, \begin{bmatrix}   var \left( \hat \alpha_{1} \right) & cov(\hat \alpha_{1},\hat \alpha_{0}) \\  cov(\hat \alpha_{1},\hat \alpha_{0}) & var \left(\hat \alpha_{0}  \right) \end{bmatrix}    \right).
\end{eqnarray}  

There are two important differences relative to the analysis from \cite{Roth}. First, we assume that  $(\hat \alpha_{1}, \hat \alpha_{0})$ are unbiased, so we can focus on the problem of spatial correlation. Second, in our setting, if there are spatially correlated shocks, then a researcher considering CRVE would be relying on  an incorrect  variance/covariance matrix for $(\hat \alpha_{1}, \hat \alpha_{0})$. We assume  that the researcher relies $\widetilde{var(\hat \alpha_\tau)} =var(\hat \alpha) -  (\mu^e_1 - \mu^e_0)'  \mathbb{E}[\Delta \xi_\tau' \Delta \xi_\tau ] (\mu^e_1 - \mu^e_0) $. Therefore, the research would rely on the correct variance matrix if $ (\mu^e_1 - \mu^e_0)'  \mathbb{E}[\Delta \xi_\tau' \Delta \xi_\tau ] (\mu^e_1 - \mu^e_0) =0$, but would underestimate the true variance if $ (\mu^e_1 - \mu^e_0)'  \mathbb{E}[\Delta \xi_\tau' \Delta \xi_\tau ] (\mu^e_1 - \mu^e_0) >0$. We can think of this normal model  as an approximation using Corollary \ref{Corollary_local}.

 By construction, if $(\mu^e_1 - \mu^e_0)'  \mathbb{E}[\Delta \xi_0' \Delta \xi_0 ] (\mu^e_1 - \mu^e_0)=0$, then pre-testing $\mathbb{E}[\hat \alpha_0]=0$ for an $5\%$ level test would reject the null  $5\%$ of the time. In contrast, if $(\mu^e_1 - \mu^e_0)'  \mathbb{E}[\Delta \xi_0' \Delta \xi_0 ] (\mu^e_1 - \mu^e_0) > 0$, then the distribution of the t-statistic would have a variance larger than one, which implies that the test would reject at a higher rate than $5\%$. An immediate consequence is that we should expect a larger fraction of applications ``surviving'' such pre-test when $(\mu^e_1 - \mu^e_0)' \mathbb{E}[\Delta \xi_0' \Delta \xi_0 ] (\mu^e_1 - \mu^e_0)$ is smaller. If we believe $ \mathbb{E}[\Delta \xi_1' \Delta \xi_1 ] \approx  \mathbb{E}[\Delta \xi_0' \Delta \xi_0 ]$, then this would also imply that the probability of surviving the pre-test would be decreasing with the degree in which $var(\hat \alpha_1)$ is underestimated. It is important to understand, however, what are the properties of the estimator for $\hat \alpha_1$ when we condition on surviving such pre-test. 

Let $B$ be the set of values for $\hat \alpha_0$ such that we fail to reject the null in the pre-test using a $t$-test  based on $\hat \alpha_0/\sqrt{\widetilde{var(\hat \alpha_0)}}$.  In this case, the pre-test is symmetric in the sense that $\hat \alpha_0$ is rejected if and only if $-\hat \alpha_0$ is rejected, even if  ${var(\hat \alpha_0)} > \widetilde{var(\hat \alpha_0)}$. The only difference is that the probability of rejecting the null for an $5\%$ level test would be 5\% if ${var(\hat \alpha_0)} =\widetilde{var(\hat \alpha_0)}$, and would be increasing in  ${var(\hat \alpha_0)} -\widetilde{var(\hat \alpha_0)}$. Therefore, from Proposition 3.1 and Corollary 3.1 from \cite{Roth} we have that $\mathbb{E}[\hat \alpha_1 | \hat \alpha_0 \in B] = \alpha$, so the DID estimator $\hat \alpha_1$ remains unbiased even if we condition on passing on such pre-test, regardless of whether there is spatial correlation. Of course, this conclusion remains valid if we consider different significance levels for the pre-test.  Moreover, since $B$ is a convex set, from Proposition 3.3 from \cite{Roth}, we also have that $var\left(\hat \alpha_1 | \hat \alpha_0 \in B \right) \leq var\left(\hat \alpha_1 \right)$.

Taken together, these results show that  pre-testing for spatial correlation can be informative about whether inference based on CRVE is reliable, and such pre-testing would not exacerbate the problem in case it fails to detect relevant spatial correlation due to noise in the data. This differs from the conclusions from \cite{Roth} when testing for pre-trends, where conditioning on passing a  pre-test for violations on parallel trends implies that the problem may be exacerbated if the parallel assumptions does not hold. If there are no spatially correlated shocks, then we should expect testing $\alpha=0$ to have  the correct level if we condition on $\hat \alpha_0 \in B$, although it may be conservative. If there are spatially correlated shocks, then conditioning on $\hat \alpha_0 \in B$ implies that we should not expect more over-rejection than if we did not consider a pre-test. 
 Moreover, if we condition on applications that pass the pre-test, then we should expect relatively fewer empirical applications in which CRVE is grossly under-estimated.

\pagebreak

\begin{figure}[H] 

\begin{center}
\caption{{\bf Simulations with the ACS - Dynamic DID}}  \label{App_fig_dynamic}

\includegraphics[scale=0.9]{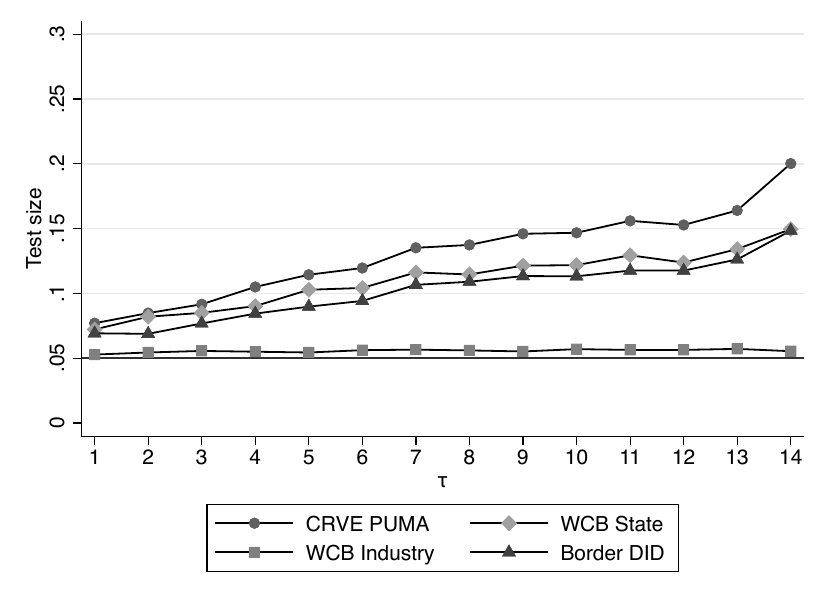}

\end{center}

\small{Notes: this figure presents rejection rates for the simulations based on the ACS data, when we consider a dynamic DID model. We present rejection rates for different inference methods for the effects $\tau$ years after the treatment. We consider the case in which treatment assignment is correlated with industry composition. Similar to the case considered in the main text (in which  size distortions are increasing with the time range $T$), in this dynamic model size distortions are increasing with the distance between the baseline and the post-treatment periods $t^\ast + \tau$.    }

\end{figure}

\pagebreak

\begin{figure}[H] 

\begin{center}
\caption{{\bf Rejection rate using \cite{Muller1,Muller2}}}  \label{Fig_MW}

\includegraphics[scale=0.9]{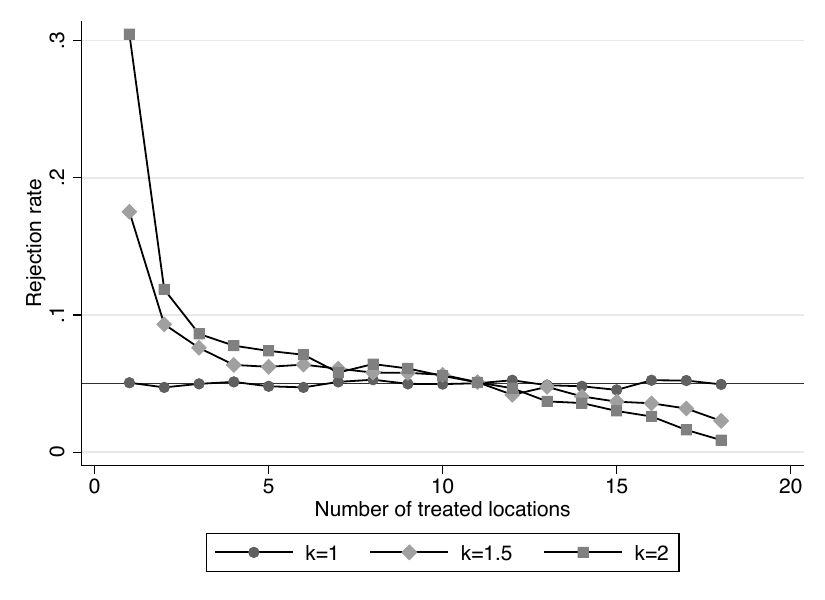}

\end{center}

\small{Notes: see details of the simulation in Appendix \ref{Appendix_alternatives_CPS}.  }

\end{figure}

\pagebreak

\begin{figure}[H] 

\begin{center}
\caption{{\bf Trends for Manufacturing vs non-Manufacturing PUMAs}}  \label{App_figure_PT}

\includegraphics[scale=0.9]{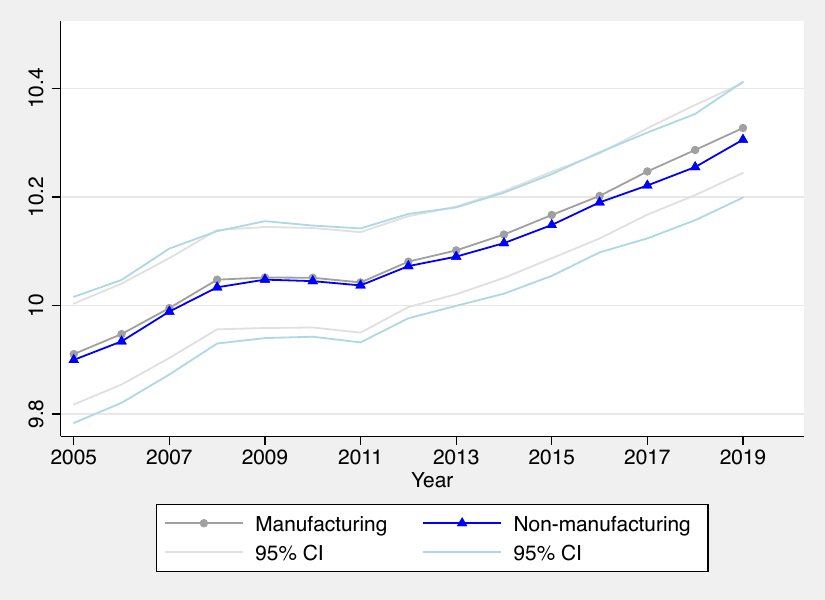}

\end{center}

\small{Notes: this figure presents the trends of log wages for PUMAs  that are relatively more concentrated into manufacturing versus those that are not. We present two-way cluster standard errors at the state and at the industry group levels. The p-value of a test that these two groups of PUMAs followed the same linear time trend is $0.468$.     }

\end{figure}

\pagebreak

\begin{figure}[H] 

\begin{center}
\caption{{\bf Trends for Different Age Groups}}  \label{App_fig_CPS}

\includegraphics[scale=0.9]{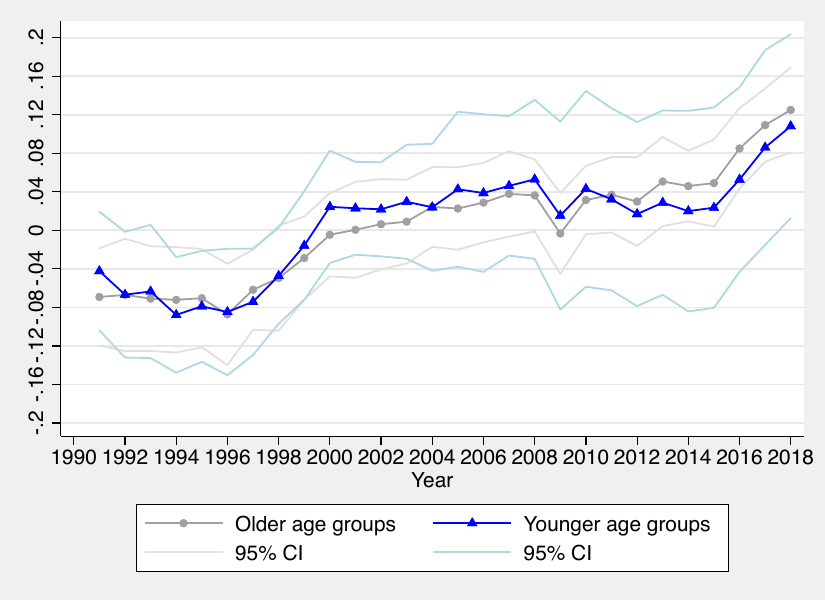}

\end{center}

\small{Notes:  this figure presents the trends of log wages for younger age groups (ages 25 to 37) and older age groups (ages 38 to 50). We subtract from both series their averages over time, so that it becomes easier to compare their trends. We present standard errors that allow for correlation between individuals within the same state, and also between individuals in close  age groups (even if  in different states).  The p-value of a test that these two groups follow the same linear trend is 0.290. 

}

\end{figure}

\pagebreak

\begin{table}[H]
  
 \begin{center}
\caption{{\bf Simulations with staggered design}} \label{Table_CCE}

\begin{tabular}{lcccccccc}
\hline
\hline

& \multicolumn{2}{c}{$F=1$} & &  \multicolumn{2}{c}{$F=2$} &  & \multicolumn{2}{c}{$F=3$} \\ \cline{2-3} \cline{5-6} \cline{8-9}

& (1) & (2) & & (3) & (4) & & (5) & (6) \\
\hline

Rejection rate \\
~~~ CCE Estimator & 0.051 & 0.048 &  & 0.050 & 0.039 &  & 0.048 & 0.542 \\
~~~ TWFE Estimator & 0.051 & 0.513 &  & 0.048 & 0.429 &  & 0.046 & 0.576 \\
 \\
Standard error \\
~~~ CCE Estimator & 0.183 & 0.156 &  & 0.176 & 0.162 &  & 0.176 & 0.590 \\
~~~ TWFE Estimator & 0.091 & 0.260 &  & 0.091 & 0.225 &  & 0.093 & 0.309 \\

 \\
Relevant spatial correlation & No & Yes &  & No & Yes &  & No & Yes \\

\hline

\end{tabular}

 \end{center}
{\small Notes:   This table presents rejection rates and standard errors (considering the standard deviation of the estimates across simulations) for the CCE  and  TWFE estimators, considering 5000 simulations with staggered designs. In all cases, we set $T=6$ and $N=300$. Units are divided into three groups. Those with $G_j = 1$ are never treated, those with $G_j = 2$ start treatment after period 2, and those with $G_j = 3$ start treatment after period 4. Outcomes are given by $Y_{jt} = \sum_{f=1}^F 0.5 \times \mu_j(f) \omega_t(f)  + \epsilon_{jt}$, where $F$ is the dimension of the linear factor model, $\omega_t(f) \sim N(0,1)$ for all $t$ and $f$, and $\epsilon_{jt} \sim N(0,1)$. All random variables are mutually independent. Each column presents different DGP's varying the dimension of the linear factor model ($F$), and the distribution of the factor loadings:

\begin{itemize}
\item  Column 1:  $F=1$; $\mu_j(1) \sim U[-1,1]$ for all $j$. Since treated and control groups have the same expected factor loading, we should not expect distortions due to the common shocks.

\item Column 2: $F=1$; $\mu_j(1) | (G = 1) \sim U[0,1]$, while $\mu_j(1) | (G \neq 1) \sim U[-1,0]$.

\item Column 3: $F=2$; $P[(\mu_j(1),\mu_j(2)) = (0,1) | G] =0.5$  and $P[(\mu_j(1),\mu_j(2)) = (1,0) | G] =0.5$. Since treated and control groups have the same expected factor loadings, we should not expect distortions due to the common shocks.

\item Column 4: $F=2$; $P[(\mu_j(1),\mu_j(2)) = (0,1) | G=1] =0.8$ and $P[(\mu_j(1),\mu_j(2)) = (1,0) | G=1] =0.2$; $P[(\mu_j(1),\mu_j(2)) = (0,1) | G\neq1] =0.2$ and $P[(\mu_j(1),\mu_j(2)) = (1,0) | G\neq1] =0.8$. Therefore, the never treated units are more likely to be affected by the first common shock.

\item Column 5: $F=3$; $P[(\mu_j(1),\mu_j(2),\mu_j(3)) = (0,0,1) | G] =P[(\mu_j(1),\mu_j(2),\mu_j(3)) = (0,1,0) | G] =P[(\mu_j(1),\mu_j(2),\mu_j(3)) = (1,0,0) | G] =1/3$.

\item Column 6: $F=3$; group $G=g$ has probability 0.8 of being  affected by the common shock $g$, and probability of 0.1 of being affected by one of the other common shocks. 

\end{itemize}

  }

\end{table}

\end{document}